\documentclass[a4paper,12pt]{article}
\usepackage[utf8]{inputenc}
\usepackage{amsthm,amsmath,amssymb,amsfonts,upgreek,tikz}
\usepackage{color}
\usetikzlibrary{arrows,positioning,cd,calc}
\usetikzlibrary{decorations.pathreplacing,decorations.markings}
\tikzset{
	on each segment/.style={
		decorate,
		decoration={
			show path construction,
			moveto code={},
			lineto code={
				\path [#1]
				(\tikzinputsegmentfirst) -- (\tikzinputsegmentlast);
			},
			curveto code={
				\path [#1] (\tikzinputsegmentfirst)
				.. controls
				(\tikzinputsegmentsupporta) and (\tikzinputsegmentsupportb)
				..
				(\tikzinputsegmentlast);
			},
			closepath code={
				\path [#1]
				(\tikzinputsegmentfirst) -- (\tikzinputsegmentlast);
			},
		},
	},
	mid arrow/.style={postaction={decorate,decoration={
				markings,
				mark=at position .5 with {\arrow[#1]{stealth}}
	}}},
}

\usepackage[pdftex, bookmarksnumbered=true]{hyperref}

\usepackage{titlesec}
\titleformat*{\section}{\Large \normalfont \bfseries}
\titleformat*{\subsection}{\large \normalfont \bfseries}
\usepackage{setspace}

\newcommand{\nc}{\newcommand}

\nc{\ny}{\nonumber}
\nc{\ra}{\rangle}
\nc{\la}{\langle}
\nc{\D}{\Delta}
\nc{\lmb}{\lambda}
\nc{\sg}{\sigma}
\nc{\ora}{\overrightarrow}
\nc{\td}{\widetilde}
\nc{\NSR}{\textsf{NSR}}
\nc{\sfF}{\textsf{F}}
\nc{\Vir}{\textsf{Vir}}
\nc{\Asl}{\widehat{\mathfrak{sl}}}
\nc{\sfV}{\textsf{V}}
\nc{\tr}{\mathrm{Tr}}
\nc{\red}{\color{red}}
\nc{\NS}{\scriptscriptstyle{\textsf{NS}}}
\nc{\R}{\scriptscriptstyle{\textsf{R}}}
\nc{\nsr}{\scriptscriptstyle{\textsf{NSR}}}
\nc{\vac}{\varnothing}
\nc{\G}{\mathsf{G}}
\nc{\BNS}{\textsf{NS}}
\nc{\BR}{\textsf{R}}
\nc{\Q}{\mathbf{Q}}
\nc{\B}{\operatorname{B}}
\nc{\ti}{\mathfrak{t}}

\numberwithin{equation}{section}
\newtheorem{thm}{Theorem}[section]
\newtheorem{prop}{Proposition}[section]
\newtheorem{lemma}{Lemma}[section]
\newtheorem{conj}{Conjecture}[section]

\newtheorem{Remark}{Remark}[section]
\newtheorem{defn}{Definition}[section]
\title{Painlev\'e equations from Nakajima-Yoshioka \\ blowup relations}
\author{M.~Bershtein, A.~Shchechkin}
\date{}
\begin{document}
	\maketitle
	\begin{abstract}
		Gamayun, Iorgov and Lisovyy in 2012 proposed that tau function of the Painlev\'e equation equals to the series of $c=1$ Virasoro conformal blocks. We study similar series of $c=-2$ conformal blocks and relate it to Painlev\'e theory. The arguments are based on Nakajima-Yoshioka blowup relations on Nekrasov partition functions. 
		
		We also study series of $q$-deformed $c=-2$ conformal blocks and relate it to $q$-Painlev\'e equation. As an application we prove formula for the tau function of $q$-Painlev\'e $A_7^{(1)'}$ equation.

	\end{abstract}
	\tableofcontents
	\section{Introduction}
	\label{sec:intro}
	\paragraph{Motivation.} The subject of the paper is the relation between conformal field theory and Painlev\'e equations (and more generally equations of isomonodromic deformation). The first example of this relation was conjectured in \cite{GIL12} and states that tau function of the Painlev\'e VI equation is equal to the (Fourier) series of $c=1$ Virasoro conformal blocks 
	\begin{equation}
	\tau(\boldsymbol{\theta};\sigma,s|z)=\sum_{n\in\mathbb{Z}} s^n  \mathcal{Z}_{c=1}(\boldsymbol{\theta};\sigma+n|z). \label{eq:GIL}
	\end{equation}
	Due to AGT relation conformal block $\mathcal{Z}_{c=1}(\boldsymbol{\theta},\sigma|z)$ is equal to 4d Nekrasov partition function, the central charge condition $c=1$ corresponds to a condition $\epsilon_1+\epsilon_2=0$ for Nekrasov parameters. The formula \eqref{eq:GIL} was proven in \cite{ILT14}, \cite{BS14}, \cite{GL16} by different methods, namely the proof in \cite{ILT14} is based on the monodromy properties of conformal blocks with degenerate field, the proof in \cite{BS14} is based on bilinear relations and proof in \cite{GL16} actually uses only combinatorial formula for Nekrasov partition function.
	
	The Painlev\'e VI equation is a particular case of the equation of the isomonodromic deformation.  There are plenty of  generalizations of the formula \eqref{eq:GIL} for the tau functions of isomonodromic deformation problems including irregular singularities and rank $N$ linear systems (see e.g. \cite{GIL13}, \cite{N15}, \cite{GM16}). 
	
%
	It is interesting to look for the analog of the formula \eqref{eq:GIL} with right side given as a series of Virasoro conformal blocks with $c\neq 1$.\footnote{It was proposed in the paper \cite{BGM17} that such series for generic central charge can be viewed as a quantum Painlev\'e tau function. We do not discuss quantization in this paper.} There are several reasons to believe that at least for special central charges the tau functions defined in such manner possess nice properties. First, the crucial step in the proof in \cite{ILT14} was the commutativity of operator valued monodromies which holds for central charges of (logarithmic extension of) minimal models $\mathcal{M}(1,n)$
	\begin{equation}\label{eq:cc:M(1,n)}
	c=1-6\frac{(n-1)^2}{n},\; n \in \mathbb{Z}\setminus \{0\}. 
	\end{equation}
	In terms of Nekrasov parameters it means that $\epsilon_1+n\epsilon_2=0$. Second, in \cite{BS14} the bilinear relations for conformal blocks were proven for any central charge and it was also observed that they could lead to the relations on tau functions for central charges \eqref{eq:cc:M(1,n)}. Third, B. Feigin in the paper \cite{F17} argued that the relation to isomonodromic deformations comes from the action of $SL(2,\mathbb{C})$ on the vertex algebra and for central charges \eqref{eq:cc:M(1,n)} such logarithmic algebras exist \cite{FGST06}.
	
	Another direction for generalization of the formula \eqref{eq:GIL} is a $q$-deformation. It was conjectured in \cite{BS16q}, \cite{JNS17} that formulas of this type solve tau form of $q$-difference Painlev\'e equations. Here one needs to replace function $\mathcal{Z}_{c=1}(\boldsymbol{\theta},\sigma|z)$ by 5d Nekrasov partition functions, or $q$-deformed conformal blocks. 
	
	In this paper we consider the case of $c=-2$, which corresponds to $\mathcal{M}(1,2)$ minimal model. Using the pure Nekrasov partition functions (4d or 5d) we introduce $c=-2$ tau functions by a (Fourier) series 
	\begin{equation}
	\uptau^{\pm}(a,s|z)=\sum_{n\in\mathbb{Z}}s^{n/2}\mathcal{Z}(a+2n\epsilon;\mp\epsilon,\pm 2\epsilon|z).    \label{eq:tau:pm}
	\end{equation}
	We study their properties and relate them to parameterless Painlev\'e III and $q$-Painlev\'e III equations. And as a by-product we prove the formula \eqref{eq:GIL} for the tau function of this $q$-Painlev\'e equation, conjectured in \cite{BS16q}.\footnote{After the paper was written, we noticed preprint \cite{MN18} with another (based on degeneration of the  results of \cite{JNS17}) proof of this conjecture .}
	
	\paragraph{Method.} The main idea is to use Nakajima-Yoshioka blowup relations on Nekrasov partition functions \cite{NY03}, \cite{NY05}, \cite{GNY06}, \cite{NY09}. They have the form 
    	\begin{equation}\label{eq:Z=ZZ}
    	\beta_D\mathcal{Z}(a,\epsilon_1,\epsilon_2|z)=\sum_{n \in \mathbb{Z}+j/2} \mathrm{D}\Big(\mathcal{Z}(a+2n\epsilon_1,\epsilon_1,-\epsilon_1+\epsilon_2|z),\mathcal{Z}(a+2n\epsilon_2,\epsilon_1-\epsilon_2,\epsilon_2|z)\Big),\quad j=0,1    
    	\end{equation}
    	where $\mathrm{D}$ is some differential (or difference) operator and $\beta_D$ is some function, possibly equal to zero, see eq. \eqref{NY}, \eqref{NY2}-\eqref{NY1}, \eqref{qNY1}-\eqref{qNY3} in the main text. 
	
	Now set $\epsilon_1+\epsilon_2=0$, then on the left side we get $c=1$ conformal block and on the right side a bilinear expression of $c=-2$ conformal blocks. Taking the sum of these relations with coefficients $s^n$ one gets $\tau(z)$ on the l.h.s. The r.h.s. can be written in terms of $c=-2$ tau functions defined in \eqref{eq:tau:pm}, namely 
	\begin{equation}
	\beta_D\tau(z)=\mathrm{D}(\uptau^+(z),\uptau^-(z)). \label{eq:tau=tau+tau-}
	\end{equation}
	Excluding $\tau(z)$ from these equations one gets closed system of bilinear relations on functions $\uptau^+(z)$, $\uptau^-(z)$. Moreover, the equations \eqref{eq:Z=ZZ} can be used to prove the bilinear relations on $\tau(z)$ (see Prop. \ref{doubleprop}, \ref{qdoubleprop})
	
	This system of bilinear relations on $\tau(z)$ is a bilinear form of Painlev\'e equations. It was shown in \cite{BS14},\cite{BS16q} that they follow from relations \begin{equation}\label{eq:Z2=ZZ}
	\tilde{\mathcal{Z}}(a,\epsilon_1,\epsilon_2|z)=\sum_{n \in \mathbb{Z}+j/2} \mathrm{D}\Big(\mathcal{Z}(a+2n\epsilon_1,2\epsilon_1,-\epsilon_1+\epsilon_2|z),\mathcal{Z}(a+2n\epsilon_2,\epsilon_1-\epsilon_2,2\epsilon_2|z)\Big),\quad j=0,1 
	\end{equation}
	after specialization $\epsilon_1+\epsilon_2=0$ and exclusion of $\tilde{\mathcal{Z}}$ (here $\tilde{\mathcal{Z}}$ is another Nekrasov partition function). The relation \eqref{eq:Z2=ZZ} is a  blowup relation for $\mathbb{C}^2/\mathbb{Z}_2$ \cite{BPSS13}. In this sense derivation of Painlev\'e equations from \eqref{eq:tau=tau+tau-} is a derivation of some  $\mathbb{C}^2/\mathbb{Z}_2$ blowup equation from ordinary $\mathbb{C}^2$ blowup equations (in case $\epsilon_1+\epsilon_2=0$).
	
	\paragraph{Content.} In Section \ref{sec:setup} we recall necessary facts about Painlev\'e equations and their solutions. In the paper we mainly consider two Painleve equations which we denote by Painlev\'e~III($D_8^{(1)}$) and $A_7^{(1)'}$ following Sakai geometric notations \cite{S01}. These equations are parameterless Painlev\'e~III equation (also called Painlev\'e~$\mathrm{III}'_3$ in the literature, see e.g. \cite{GIL13}) and its $q$-deformation. These equations are written in terms of two tau functions $\tau$ and $\tau_1$ related by B\"acklund transformation. The precise statement about solutions in terms of $\epsilon_1+\epsilon_2=0$ Nekrasov functions is given in Theorem \ref{Kievthm}, which was proven in \cite{ILT14},\cite{BS14},\cite{GL16} and Theorem \ref{qKievthm} which we prove in Section \ref{sec:qtau}.
	
	In Section \ref{sec:tau} we discuss differential equations and 4d Nekrasov partition functions. We construct $\uptau^+$ and $\uptau^-$ and their B\"acklund transformations $\uptau_1^+$ and $\uptau_1^-$. The bilinear equations, which follow from the Nakajima-Yoshioka blowup relations take the form
	\begin{eqnarray}
	\uptau^+\uptau^-=\tau,\quad D^1_{[\log z]}(\uptau^+,\uptau^-)=z^{1/4}\tau_1, \quad
	D^2_{[\log z]}(\uptau^+,\uptau^-)=0,\\
	D^3_{[\log z]}(\uptau^+,\uptau^-)=z^{1/4}\left(z\frac{d}{dz}\right)\tau_1,\quad D^4_{[\log z]}(\uptau^+,\uptau^-)=-2z\tau. 
	\end{eqnarray}
	Also we have B\"acklund transformed versions of these equations. Here $D^k_{[\log z]}$ denotes Hirota differential operator. We also find definition of $\uptau^+$ and $\uptau^-$ in terms of Painlev\'e theory, namely 
	\begin{equation}\label{eq:tau+-KZ}
	z\frac{d}{dz}\uptau^{\pm}=\frac12(\zeta\mp i\sqrt{\zeta'})\uptau^{\pm}, 
	\end{equation}
	where $\zeta$ is non-autonomous Hamiltonian of Painlev\'e equation \eqref{zeta3} (see e.g. \cite{GIL13}). It is instructive to compare this equation with a defining equation on Painlev\'e tau function $z\frac{d\tau}{dz}=\zeta\tau$. From the CFT point of view the  equation \eqref{eq:tau+-KZ} should be a Knizhnik–Zamolodchikov equation for the conformal field theory with central charge~$c=-2$ and $SL(2,\mathbb{C})$ action, this theory is symplectic fermion theory, see e.g.~\cite{K00}. 
	
	In Section \ref{sec:qtau} we discuss $\uptau^+$ and $\uptau^-$ in $q$-difference setting. First, in Section \ref{ssec:qtaum2}, from Nakajima-Yoshioka blowup relation we obtain 
	equations 
	\begin{eqnarray}\label{eq:tau+-q:eq}
	\uptau^+\uptau^-=\tau,\quad \overline{\uptau^+}\underline{\uptau^-}+\underline{\uptau^+}\overline{\uptau^-}=2\tau, \quad
	\overline{\uptau^+}\underline{\uptau^-}-\underline{\uptau^+}\overline{\uptau^-}=-2z^{1/4}\tau_1.
	\end{eqnarray} 
	Here we use notations $\overline{f(z)}=f(qz), \underline{f(z)}=f(q^{-1}z)$. We deduce from equations \eqref{eq:tau+-q:eq} bilinear equations on functions $\tau,\tau_1$, thus proving Theorem \ref{qKievthm}. In Section \ref{ssec:CS} we use the same idea for 5d $SU(2)$ Nekrasov functions with Chern-Simons term at the level $m=1,2$. We show that tau functions for $m=2$ again solve $q$-Painlev\'e $A_7^{(1)'}$ and for $m=1$ they  solve $q$-Painlev\'e $A_7^{(1)}$ in agreement with a conjecture of \cite{BGM18}.
	
	It is natural to ask what is the meaning of the dynamics \eqref{eq:tau+-q:eq} for functions $\uptau^+,\uptau^-, \uptau_1^+, \uptau_1^-$. In Section \ref{ssec:clust2} we show that this is a particular case of another $q$-Painlev\'e, namely $A_3^{(1)}$ equation, which is $q$-Painlev\'e VI. Jimbo, Nagoya, Sakai in \cite{JNS17} proposed the formula for tau functions of this equation, this formula has the form of Fourier series \eqref{eq:GIL} of 5d Nekrasov partition function with four matters. Therefore we get a surprising relation between $c=-2$ pure Nekrasov partition function and $c=1$ Nekrasov partition function with special values of masses of matter fields, see Conjecture \ref{conj:prdx}. 
	
	Bonelli, Grassi, Tanzini in the paper \cite{BGT17} observed that the function \eqref{eq:GIL} has a natural meaning in framework of topological strings/spectral theory duality \cite{GHM14}. Namely for $|q|=1$, $s=1$ and appropriate choice of $\mathcal{Z}$ the tau function conjecturally equals to the spectral determinant $\Xi$ up to some simple factor. Moreover for $z=q^M,\,M\in\mathbb{Z}$ the function $\Xi$ becomes grand canonical partition function of ABJ theory. And in this case there exists natural factorization $\Xi=\Xi^+\Xi^-$, moreover it was proposed in \cite{GHM14'} that functions $\Xi^+,\Xi^-$ satisfy additional, so called Wronskian-like, relations. We show in Section \ref{ssec:ABJM} that $\Xi^+$, $\Xi^-$ should be equal to our $\uptau^+$, $\uptau^-$ up to some simple factor. Probably, this means the existence of generalization of topological strings/spectral theory duality for the case of $q=t^2$ on topological strings side. So we see that $c=-2$ tau functions appear naturally in this approach too (and hopefully in any approach to CFT/isomonodromy correspondence).
	
	We conclude with the several directions for future study and partial results in Discussion.

	\paragraph{Acknowledgements.}
	We thank  B.~Feigin, P.~Gavrylenko, A.~Grassi, A.~Marshakov, A.~Mi\-ro\-nov, H.~Nakajima for interest to our work and discussions. We are grateful to the referee for valuable comments.
	
	This work is partially supported by HSE University Basic Research Program and funded (partially) by the Russian Academic Excellence Project '5-100' and by the RFBR grant mol\_a\_ved 18-31-20062. Authors were also supported in part by Young Russian Mathematics award. The work of M.B. in Landau Institute has been funded by FANO assignment 0033-2018-0006.
	
	\section{Painlev\'e equations and $c=1$ tau functions}
	\label{sec:setup}
	
	\subsection{Painlev\'e III($D_8^{(1)}$) and $A_7^{(1)'}$  equations}
	\label{ssec:setupeq}
	We recall several facts about one of the simplest Painlev\'e equations --- Painlev\'e III($D_8^{(1)}$) (or Painlev\'e $\mathrm{III}'_3$) 
	and its $q$-deformation --- Painlev\'e $A_7^{(1)'}$ equation following \cite{BS16b}, \cite{BS16q} and references therein.
	
	\paragraph{Painlev\'e III($D_8^{(1)})$ equation.}
	
	We will start from the so-called $\zeta$-form of 
	this equation
	\begin{equation}
	(z\ddot\zeta(z))^2=4\dot\zeta(z)^2(\zeta(z)-z\dot\zeta(z))-4\dot\zeta(z),\quad \ddot\zeta(z)\neq 0, \label{zeta3}
	\end{equation}
	where we use dot notation for differentiation with respect to $z$. This equation is also called radial sine-Gordon equation (see e.g. \cite[Eq. (2.6)]{BS16b}).
	
	This equation has a $\mathbb{Z}_2$ symmetry $\pi$: $z\mapsto z,\,\zeta\mapsto \zeta_1$, where $\dot{\zeta}\dot{\zeta_1}=z^{-1}$ and
	constant of integration for $\zeta_1$ is defined by the fact that $\zeta_1$ satisfies \eqref{zeta3}.
	This transformation is called B\"acklund transformation of given  Painlev\'e equation.
	We will mark B\"acklund transformed variables by the subscript ${}_1$. Also, we will optionally mark initial variables by the subscript ${}_0$ where it is convenient.
	
	We introduce the tau function, defined up to a $z$-constant, by the formula
	\begin{equation}
	\zeta(z)=\frac{d \log\tau(z)}{d\log z} \quad \textrm{and conversely} \quad \tau=\exp\left(\int\zeta(z) d\log z\right) .\label{tauzeta}
	\end{equation}
	It is convenient to write any Painlev\'e equation as an equation(s) on tau function. It appears that such
	equations are bilinear in $\tau$ and it is convenient to write them using Hirota differential operators $D^k_{[x]}$
	\begin{equation}
	f(e^{\alpha}z)g(e^{-\alpha}z)=\sum\limits_{k=0}^{\infty}D^{k}_{[\log z]}(f(z),g(z))\frac{\alpha^k}{k!}.\label{Hiru}
	\end{equation}
	The first examples of Hirota operators are
	\begin{equation}
	D^0_{[\log z]}(f(z),g(z))=f(z)g(z),\qquad D^1_{[\log z]}(f(z),g(z))=f'(z)g(z)-f(z)g'(z), 
	\end{equation}
	where we use prime notation for differentiation with respect to $\log z$.
	In this paper we use only Hirota derivatives with respect to the logarithm of a variable.
	
	For Painlev\'e III($D_8^{(1)}$), the equation on the single tau function $\tau$ is a rather cumbersome equation of order $4$ (see \cite[Eq. (4.7)]{BS14}).
	We will use a more simple tau-form: the so-called Toda-like equations on the tau function and its B\"acklund transformation $\tau_1$
	\begin{equation}
	\begin{aligned}
	D^2_{[\log z]}(\tau,\tau)&=-2z^{1/2}\tau_1^2, \\
	D^2_{[\log z]}(\tau_1,\tau_1)&=-2z^{1/2}\tau^2, \label{eq:Toda:old}
	\end{aligned}
	\end{equation}
	where we used appropriate constant normalization of $\tau$ and $\tau_1$ to make two equations look the same.
	These equations are almost equivalent to \eqref{zeta3}.
	Namely, we have the freedom of multiplying the solutions of \eqref{eq:Toda:old} by $z^K$ as $\tau\mapsto z^K \tau, \tau_1\mapsto z^K \tau_1$.
	This freedom corresponds via \eqref{tauzeta} to the freedom to add a constant $K$ to $\zeta$ and  $\zeta_1$. 
	The above equivalence is exactly up to this freedom. 
	\begin{prop}\cite{BS16b}
		The functions $\zeta$ and $\zeta_1$ are solutions of \eqref{zeta3} if and only if the corresponding functions $\tau$ and $\tau_1$ multiplied on $z^K$ (with certain $K$) are solutions of \eqref{eq:Toda:old}
		\label{prop:Toda}
	\end{prop}
	This freedom and constant $K$ will appear also in the Subsection \ref{ssec:KZdif}.
	
	It is convenient to write Toda-like equations as
	\begin{equation}
	D^2_{[\log z]}(\tau_j,\tau_j)=-2z^{1/2}\tau_{j+1}\tau_{j-1},\quad j\in \mathbb{Z}/2\mathbb{Z}. \label{eq:Toda}
	\end{equation}
	
	\paragraph{Painlev\'e $A_7^{(1)'}$ equation.}
	This is a second order $q$-difference equation on $G(z)$
	\begin{equation}
	\overline{G}\underline{G}=\frac{(G-z)^2}{(G-1)^2}, \label{eq:qG}
	\end{equation}
	where $G=G(z)$ and we denote $\overline{f(z)}=f(qz), \underline{f(z)}=f(q^{-1}z)$ for the arbitrary function $f(z)$. 
	Note that this equation is not a direct $q$-analog of $\zeta$-equation \eqref{zeta3} but a $q$-analog of the "standard" coordinate form
	of the Painlev\'e III($D_8^{(1)}$) equation.
	Equation \eqref{eq:qG} also has $\mathbb{Z}_2$ B\"acklund symmetry $\pi$:$z\mapsto z, G(z)\mapsto z/G(z)$.
	
	Below we will use only Toda-like tau form of equation \eqref{eq:qG} \cite{BS16q}
	\begin{equation}
	\overline{\tau_j}\underline{\tau_j}=\tau_j^2-z^{1/2}\tau_{j+1}\tau_{j-1}, \quad j\in \mathbb{Z}/2\mathbb{Z}, \label{eq:qToda}
	\end{equation}
	Connection between $G(z)$ and $\tau_j$ are given by
	\begin{prop}\cite{BS16q}
		If the functions $\tau_j(z)$, $j\in \mathbb{Z}/2\mathbb{Z}$ satisfy \eqref{eq:qToda}, then $G(z)=z^{1/2}\frac{\tau_0^2}{\tau_1^2}$
		satisfies \eqref{eq:qG}. 
	\end{prop}
	For the proof, substitute $G(z)$ to \eqref{eq:qG} and get
	\begin{equation}
	\left(\frac{\overline{\tau_0}\underline{\tau_0}}{\overline{\tau_1}\underline{\tau_1}}\right)^2=\left(\frac{\tau_0^2-z^{1/2}\tau_1^2}{\tau_1^2-z^{1/2}\tau_0^2}\right)^2.
	\end{equation}
	\begin{Remark}
		Naively, the substitution in the proof of this Proposition could be used for the proof of inverse statement:
		if we take a solution $G(z)$ of \eqref{eq:qG} then there exist $\tau_0$ and $\tau_1$ satisfying \eqref{eq:qToda} such that
		$G(z)=z^{1/2}\frac{\tau_0^2}{\tau_1^2}$. Indeed, it follows from the previous equation that there
		exists some function $f(z)$ such that
		\begin{equation}
		\overline{\tau_j}\underline{\tau_j}=f(z)(\tau_j^2-z^{1/2}\tau_{j+1}\tau_{j-1}), \quad j\in \mathbb{Z}/2\mathbb{Z},
		\end{equation}
		where we omit the sign which could appear when the square root was applied.
		We have the freedom to multiply $\tau_j$ by the same function $\lmb$ in the $G(z)=-z^{1/2}\frac{\tau_0^2}{\tau_1^2}$.
		Taking this function $\lmb$ satisfying $\frac{\overline{\lmb}\underline{\lmb}}{\lmb^2}=f(z)$
		we can remove $f(z)$. These are of course rough arguments neglecting the questions of the analytic properties of the tau functions.
	\end{Remark}

	Continuous limit from \eqref{eq:qToda} to \eqref{eq:Toda} is given by $z^{1/2}\mapsto R^2 z^{1/2}$, $q=e^{-R}$ and $R\rightarrow 0$.
	
	Note that we will use in this paper the same notations for tau functions, $z$-variable and some other objects related both to continuous and
	$q$-deformed case. 
	We hope this would not lead to confusion.
	
	\paragraph{Algebraic solutions of Painlev\'e III($D_8^{(1)}$) and $A_7^{(1)'}$ equations.}
	Painlev\'e III($D_8^{(1)}$) equation has only two algebraic solutions \cite{Gr84}. These are the only solutions invariant under the B\"acklund
	transformation $\pi$. The corresponding tau functions are
	\begin{equation}\label{eq:algebraic}
	\tau(z)=\tau_1(z)=z^{1/16}e^{\mp 4\sqrt{z}}. 
	\end{equation}
	These algebraic solutions can be $q$-deformed to the algebraic solutions of Painlev\'e $A_7^{(1)'}$ equation \cite{BS16q}.
	These solutions are also invariant under B\"acklund transformation $\pi$. I.e. they satisfy $G(z)=z/G(z)$, hence $G(z)=\pm\sqrt{z}$, i.e. the term "algebraic" corresponds to $G(z)$. 
	The corresponding tau functions are given by
	\begin{equation}\label{eq:q-algebraic}
	\tau=\tau_1=z^{1/16}(\pm q^{1/2}z^{1/2};q^{1/2},q^{1/2})_{\infty},
	\end{equation}
	where the double $q$-Pochhammer symbol $(\cdot;\cdot,\cdot)_{\infty}$ is defined in the formula \eqref{Pochhammer_def}.

	\subsection{Power series representation for the tau function}
	\label{ssec:setuptau}
	
	Power series for the tau function will be Fourier series consisting of  pure gauge SUSY $SU(2)$ partition functions --- 4d and 5d for the cases of
	Painlev\'e $D_8^{(1)}$ and $A_7^{(1)'}$ equations correspondingly. These partition functions $\mathcal{Z}$ split into three factors (we follow conventions of \cite{NY03L}, \cite{NY05})
	\begin{equation}
	\mathcal{Z}=\mathcal{Z}_{cl}\mathcal{Z}_{1-loop}\mathcal{Z}_{inst}.\label{Zstr}
	\end{equation}
	In loc. cit. $\mathcal{Z}_{cl}$ and $\mathcal{Z}_{1-loop}$ appear from the so-called "perturbative" part. More details about 
	$\mathcal{Z}_{cl}$ and $\mathcal{Z}_{1-loop}$ and connection with conventions from loc. cit. one can find in the Appendix~ \ref{sec:perturb}.
	
	\paragraph{Four-dimensional Nekrasov functions.}
	The 4d partition function depend on parameters of the $\Omega$-background $\epsilon_1, \epsilon_2$ and vacuum expectation values $a_1, a_2$
	with condition $a_1+a_2=0$ (we denote $a=a_1-a_2$). Then different factors of the 4d partition function are given by the formulas
	\begin{equation}
	\begin{aligned}
	\mathcal{Z}_{cl}(a;\epsilon_1,\epsilon_2|\Lambda)&=\Lambda^{-\frac{a^2}{\epsilon_1\epsilon_2}},\\
	\mathcal{Z}_{1-loop}(a;\epsilon_1,\epsilon_2)&=\exp(-\gamma_{\epsilon_1,\epsilon_2}(a;1)-\gamma_{\epsilon_1,\epsilon_2}(-a;1)), \\
	\mathcal{Z}_{inst}(a;\epsilon_1,\epsilon_2|\Lambda)&=\sum_{\lmb^{(1)},\lmb^{(2)}}\frac{(\Lambda^4)^{|\lmb^{(1)}|+|\lmb^{(2)}|}}
	{\prod_{i,j=1}^2\mathsf{N}_{\lmb^{(i)},\lmb^{(j)}}(a_i-a_j;\epsilon_1,\epsilon_2)}, \quad |\lmb|=\sum\lmb_j, \\
	\mathsf{N}_{\lmb,\mu}(a;\epsilon_1,\epsilon_2)&=\prod_{s\in\mathbb{\lmb}}(a-\epsilon_2(a_{\mu}(s)+1)+\epsilon_1l_{\lmb}(s))\prod_{s\in\mathbb{\mu}}(a+\epsilon_2a_{\lmb}(s)-\epsilon_1(l_{\mu}(s)+1)),
	\label{4d_def}
	\end{aligned}
	\end{equation}
	where $\lmb^{(1)}, \lmb^{(2)}$ are partitions, $a_{\lmb}(s), l_{\lmb}(s)$ denote the lengths of arms and legs for the box $s$ in the
	Young diagram corresponding to the partition $\lmb$. 
	The function  $\gamma_{\epsilon_1,\epsilon_2}(x;\Lambda)$ is defined in the Appendix \ref{sec:q} by the formula \eqref{gammaNYdef}.
	The function $\mathcal{Z}_{inst}(a;\epsilon_1,\epsilon_2|\Lambda)$ satisfies elementary symmetry properties\footnote{The symmetry $\epsilon_1\leftrightarrow\epsilon_2$
		however is not obvious from the definition \eqref{4d_def} but it follows from the general construction of the instanton partition function.}
	\begin{equation}
	\mathcal{Z}_{inst}(a;\epsilon_1,\epsilon_2|\Lambda)=\mathcal{Z}_{inst}(a;\epsilon_2,\epsilon_1|\Lambda)=\mathcal{Z}_{inst}(-a;\epsilon_1,\epsilon_2|\Lambda)=\mathcal{Z}_{inst}(a;-\epsilon_1,-\epsilon_2|\Lambda). \label{Zsymm} 
	\end{equation}
	The symmetries $\epsilon_1\leftrightarrow \epsilon_2$, $a\mapsto -a$ are also satisfied by the full Nekrasov function $\mathcal{Z}$. However, for the
	symmetry $\epsilon_1, \epsilon_2 \mapsto -\epsilon_1,-\epsilon_2$ we have invariance of the $\mathcal{Z}_{cl}$ but 
	\begin{equation}
	\mathcal{Z}_{1-loop}(a;-\epsilon_1,-\epsilon_2)=-\frac{\sin(\pi a/\epsilon_2)}{\sin(\pi a/\epsilon_1)}\mathcal{Z}_{1-loop}(a;\epsilon_1,\epsilon_2),\quad \operatorname{Re}\epsilon_1<0, \operatorname{Re}\epsilon_2>0, \label{loopassym}
	\end{equation} 
	where we used \eqref{gammatrans}, \eqref{gammashift}, \eqref{gamma1trans}, \eqref{gamma1scal}, \eqref{gamma1} successively. So symmetry is broken unless $\epsilon_1 + \epsilon_2 = 0$.
	
	\paragraph{Five-dimensional Nekrasov functions.}
	In the 5d case, this partition function also depends on the radius $R$ of the 5th compact dimension.
	It is convenient to write parameters as $u_i=e^{Ra_i}$, $q_i=e^{R\epsilon_i}, i=1,2$ with condition $u_1u_2=1$  (we denote $u=u_1/u_2$).
	Then different factors in the definition of the 5d partition function are given by
	\begin{equation}
	\begin{aligned}
	\mathcal{Z}_{cl}(u;q_1,q_2|\Lambda)&=(q_1^{-1}q_2^{-1}\Lambda^4)^{-\frac{\log^2 u_1+\log^2 u_2}{2\log q_1\log q_2}},\\
	\mathcal{Z}_{1-loop}(u;q_1,q_2)&=(u_1/u_2;q_1,q_2)_{\infty}(u_2/u_1;q_1,q_2)_{\infty}, \\
	\mathcal{Z}_{inst}(u;q_1,q_2|\Lambda)&=\sum_{\lmb^{(1)},\lmb^{(2)}}\frac{(q_1^{-1}q_2^{-1}\Lambda^4)^{|\lmb^{(1)}|+|\lmb^{(2)}|}}
	{\prod_{i,j=1}^2\mathsf{N}_{\lmb^{(i)},\lmb^{(j)}}(u_i/u_j;q_1,q_2)},\\
	\mathsf{N}_{\lmb,\mu}(u;q_1,q_2)&=\prod_{s\in\mathbb{\lmb}}\left(1-uq_2^{-a_{\mu}(s)-1} q_1^{l_{\lmb}(s)}\right)\prod_{s\in\mathbb{\mu}}\left(1-uq_2^{a_{\lmb}(s)} q_1^{-l_{\mu}(s)-1}\right).
	\label{5d_def}
	\end{aligned}
	\end{equation}
	After rescaling $\Lambda^2\mapsto R^2 \Lambda^2$ in the limit $R\rightarrow 0$ 5d Nekrasov function goes to its 4d analog.
	The function $\mathcal{Z}_{inst}(a;q_1,q_2|\Lambda)$ satisfies elementary symmetry properties:
	\begin{equation}
	\mathcal{Z}_{inst}(u;q_1,q_2|\Lambda)=\mathcal{Z}_{inst}(u;q_2,q_1|\Lambda)=\mathcal{Z}_{inst}(u^{-1};q_1,q_2|\Lambda)=\mathcal{Z}_{inst}(u;q_1^{-1},q_2^{-1}|\Lambda). \label{qZsymm} 
	\end{equation}
	The symmetries $q_1\leftrightarrow q_2$, $u\mapsto u^{-1}$ are also satisfied by the full Nekrasov function $\mathcal{Z}$. However, for the
	symmetry $q_1, q_2\mapsto q_1^{-1}, q_2^{-1}$ we have
	\begin{equation}
	\mathcal{Z}_{1-loop}(u;q_1^{-1},q_2^{-1})=-u\theta^{-1}(u;q_1)\theta^{-1}(u;q_2) \mathcal{Z}_{1-loop}(u;q_1,q_2) \label{qloopassym}
	\end{equation} 
	(where we used \eqref{qtrans}, \eqref{qshift}, \eqref{eq:theta} successively) and
	\begin{equation}
	\mathcal{Z}_{cl}(u;q_1^{-1},q_2^{-1}|\Lambda)=(q_1q_2)^{-\frac{\log^2 u}{2\log q_1\log q_2}}\mathcal{Z}_{cl}(u;q_1,q_2|\Lambda), \label{qclassym}
	\end{equation}
	so symmetry is broken for all cases except $q_1q_2=1$.
	
	In this paper both for 5d and 4d case we will consider only region $\operatorname{Re}\epsilon_1\lessgtr0, \operatorname{Re}\epsilon_2\gtrless 0$.
	Both for 5d and 4d cases, if  ${\epsilon_1}/{\epsilon_2}\in\mathbb{Q}$, then one can analogously
	to the \cite[Prop. 3.1.]{BS16q} prove that power series in $z$ for $\mathcal{Z}_{inst}$ converges uniformly and absolutely
	on every bounded subset of~$\mathbb{C}$.
	
	\paragraph{CFT notations.}
	These partition functions via AGT relation correspond to the Whittaker limits of the Virasoro (for 4d) \cite{G09} or $q$-Virasoro (for 5d) \cite{AY09}, \cite{Y14} conformal blocks 
	with central charge $c$ and conformal weight $\D$ given by the formulas
	\begin{equation}
	c=1+6\frac{(\epsilon_1+\epsilon_2)^2}{\epsilon_1\epsilon_2}, \quad \D=\frac{(\epsilon_1+\epsilon_2)^2-a^2}{4\epsilon_1\epsilon_2}. \label{cD}
	\end{equation} 
	Note that 4d partition function $\mathcal{Z}$ depend only on ratios of $a_1, a_2, \epsilon_1, \epsilon_2, \Lambda$ so it will be convenient to move to
	these CFT notations without such scaling. Additionaly to the formulas \eqref{cD} we also denote
	\begin{equation}
	\sg=-\frac{a}{2\epsilon_1}, \quad z=\frac{\Lambda^4}{\epsilon_1^2\epsilon_2^2}. \label{CFT_not_aux}
	\end{equation}
	So for the 4d partition function we change notations $\mathcal{Z}(a;\epsilon_1,\epsilon_2|\Lambda)\rightarrow \mathcal{Z}_{c=\ldots}(\sg|z)$.
	In the 5d case, we denote $z=\Lambda^4$ and occasionally use the notation $\sg$ in this case.
	
	\paragraph{Tau function as a series of the partition functions.}
	We are now ready to formulate the statement already announced at the beginning of the subsection.
	Introduce a tau function $\tau(\sg,s|z)$ corresponding to some partition function $\mathcal{Z}(\sg|z)$ as a series 
	\begin{equation}
	\tau(\sg,s|z)=\sum_{n\in\mathbb{Z}}s^n \mathcal{Z}(\sg+n|z).\label{Kiev}
	\end{equation}
	This function satisfies elementary properties
	\begin{equation}
	\tau(\sg+1,s|z)=s^{-1}\tau(\sg,s|z), \quad \tau(-\sg,s^{-1}|z)=\tau(\sg,s|z), \label{Kievsymm}
	\end{equation}
	where the second property is due to the relation $\mathcal{Z}(-\sg|z)=\mathcal{Z}(\sg|z)$, which holds for all partition functions we consider.
	
	The following result was proposed in \cite{GIL12}, \cite{GIL13} and it was proved in \cite{BS14}, \cite{ILT14}, \cite{GL16} in different ways.
	\begin{thm}\label{Kievthm}
		Tau function $\tau(\sg,s|z)$ of Painlev\'e III($D_8^{(1)}$) equation is given by the formula \eqref{Kiev} with $\mathcal{Z}(\sg|z)=\mathcal{Z}_{c=1}(\sg|z)$. 
	\end{thm}
	Partition function $\mathcal{Z}_{c=1}$ corresponds to the case $\epsilon_1+\epsilon_2=0$ via the first equation of \eqref{cD}. 
	For the Painlev\'e III($D_8^{(1)}$) equation $\sg,s$ play role of the integration constants.
	Note that 1-loop term in this case is usually written in terms of Barnes $\mathsf{G}$-function as 
	\begin{equation}
	\mathcal{Z}_{1-loop}(\sg)=\frac{1}{\mathsf{G}(1-2\sg)\mathsf{G}(1+2\sg)}.
	\end{equation}
	See the last paragraph (formula \eqref{loopG}) of Appendix \ref{sec:q} for deriving this formula from \eqref{4d_def}.
	
	Note that the power series representation for $\tau_1$ is given by the formula slightly different from \eqref{Kiev}:
	\begin{equation}
	\tau_1(\sg,s|z)=\sum_{n\in\mathbb{Z}+1/2}s^n \mathcal{Z}_{c=1}(\sg+n|z)=s^{1/2}\tau(\sg+1/2,s|z),
	\end{equation}
	the only difference is in the region 
	of summation. Therefore equations \eqref{eq:Toda} could be rewritten as a single equation on $\tau(\sg,s|z)$
	\begin{equation}
	D^2_{[\log z]}(\tau(\sg,s|z),\tau(\sg,s|z))=-2z^{1/2}\tau(\sg+1/2,s|z)\tau(\sg-1/2,s|z). \label{eq:Todasg} 
	\end{equation}

	In the work \cite{BS16q}, a $q$-deformation of the Theorem \ref{Kievthm} was conjectured.
	We will present the proof of this result in Subsection
	\ref{ssec:qtaum2}
	(see Proposition \ref{qdoubleprop}). For another proof, see \cite{MN18}.
	\begin{thm}\label{qKievthm}
		Consider the tau function given by the formula \eqref{Kiev}, with $\mathcal{Z}(\sg|z)$ 
		given by the 5d pure gauge SUSY $SU(2)$ partition function \eqref{5d_def} with $\epsilon_1+\epsilon_2=0$ and $e^{R\epsilon_2}=q$. Then this tau function is a tau function of Painlev\'e ($A_7^{(1)'}$) equation.\label{qtauconj}
	\end{thm}
	Equations \eqref{eq:qToda} could also be rewritten as a single equation on $\tau(u,s|z)$
	\begin{equation}
	\tau(u,s|qz)\tau(u,s|q^{-1}z)=\tau^2(u,s|z)-z^{1/2}\tau(uq,s|z)\tau(uq^{-1},s|z) \label{eq:qTodasg}.
	\end{equation}
	
	\begin{Remark}\label{CRemark}
		Note that this theorem holds not only when $\mathcal{Z}_{cl}$ is defined by \eqref{5d_def} but for
		arbitrary $\mathcal{Z}_{cl}=C(u;q|z)$ which satisfies \cite[(3.7)-(3.9)]{BS16q}.
	\end{Remark}


	\section{From blowup relations to $c={-}2$ tau functions}
	\label{sec:tau}
	\subsection{Definition of $c={-}2$ tau functions}
	\label{ssec:taum2}
	
	The functions $\mathcal{Z}(a;\epsilon_1,\epsilon_2|\Lambda)$ are known to satisfy Nakajima-Yoshioka blowup relations \cite[(6.13)]{NY03} 
	(see also \cite[(5.3)]{BFL13} for CFT interpretation). We write them in terms of the full partition functions as in \cite[(5.2)]{NY03L}
	\begin{equation}
	\mathcal{Z}(a;\epsilon_1,\epsilon_2|\Lambda)=\sum_{n\in\mathbb{Z}} \mathcal{Z}(a+2\epsilon_1 n;\epsilon_1,\epsilon_2-\epsilon_1|\Lambda)\mathcal{Z}(a+2\epsilon_2 n;\epsilon_1-\epsilon_2,\epsilon_2|\Lambda).
	\label{NY}
	\end{equation}
	Imposing condition $\epsilon_1+\epsilon_2=0$ we get in the CFT notations
	\begin{equation}
	\mathcal{Z}_{c=1}(\sg|z)=\sum_{n\in\mathbb{Z}}\mathcal{Z}^+_{c=-2}\left(\sg-n\left|\frac{z}4\right.\right) 
	\mathcal{Z}^-_{c=-2}\left(\sg+n\left|\frac{z}4 \right.\right), \label{NYs}
	\end{equation}
	where $\mathcal{Z}^{\pm}_{c=-2}(\sg|z)=
	\mathcal{Z}(\sg;\pm\epsilon_1, \mp2\epsilon_1|z)$.
	Only in the case of $c=1$ in the l.h.s. we obtain a product of the partition functions with the same ($c=-2$) central charges in r.h.s.
	Partition functions $\mathcal{Z}^{\pm}_{c=-2}$ are not equal due to an asymmetry of the $1$-loop factor \eqref{loopassym}.
	However, this asymmetry is expressed as
	\begin{equation}
	\mathcal{Z}^+_{c=-2}\left(\sg\left|\frac{z}4\right.\right) =\frac{1}{2\cos\pi\sg}\mathcal{Z}^-_{c=-2}\left(\sg\left|\frac{z}4 \right.\right),\quad \operatorname{Re}\epsilon_1>0\label{pmsym}
	\end{equation}
	and it will be useful to introduce the combination
	\begin{equation}
	\mathcal{Z}_{c=-2}\left(\sg\left|\frac{z}4\right.\right)=(2\cos\pi\sg)^{1/2}\mathcal{Z}^+_{c=-2}\left(\sg\left|\frac{z}4\right.\right).
	\end{equation}
	
	It is natural to make discrete time Fourier transform of  \eqref{NYs} to obtain in the l.h.s. a tau function of type \eqref{Kiev}.
	In the r.h.s. we obtain
	\begin{equation}
	\begin{aligned}
	&\sum_{n_1,n_2\in\mathbb{Z}}s^{n_1}\mathcal{Z}^+_{c=-2}\left(\sg+n_1-n_2\left|\frac{z}4\right.\right) \mathcal{Z}^-_{c=-2}\left(\sg+n_1+n_2\left|\frac{z}4\right.\right)=\\
	&=\sum_{n_1,n_2\in\mathbb{Z}|n_1+n_2\in 2\mathbb{Z}}+\sum_{n_1,n_2\in\mathbb{Z}|n_1+n_2\in 2\mathbb{Z}+1}=
	\left|\left|n_{\pm}=\frac12(n_1\pm n_2)\right|\right|=\\
	&=\sum_{n_+\in\mathbb{Z}}s^{n_+}\mathcal{Z}^+_{c=-2}\left(\sg+2n_+\left|\frac{z}4\right.\right)\sum_{n_-\in\mathbb{Z}}s^{n_-}\mathcal{Z}^-_{c=-2}\left(\sg+2n_-\left|\frac{z}4\right.\right)+\\
	&+\sum_{n_+\in\mathbb{Z}+1/2}s^{n_+}\mathcal{Z}^+_{c=-2}\left(\sg+2n_+\left|\frac{z}4\right.\right)\sum_{n_-\in\mathbb{Z}+1/2}s^{n_-}\mathcal{Z}^-_{c=-2}\left(\sg+2n_-\left|\frac{z}4\right.\right)=\\
	&=\sum_{n_+\in\mathbb{Z}}s^{n_+/2}\mathcal{Z}^+_{c=-2}\left(\sg+n_+\left|\frac{z}4\right.\right)\sum_{n_-\in\mathbb{Z}}s^{n_-/2}\mathcal{Z}^-_{c=-2}\left(\sg+n_-\left|\frac{z}4\right.\right),\label{Ftrcalc}
	\end{aligned}
	\end{equation}
	where the last equality follows from 
	\begin{equation}
	\mathcal{Z}^+_{c=-2}(\sg+2n_++1) \mathcal{Z}^-_{c=-2}(\sg+2n_-)+\mathcal{Z}^-_{c=-2}(\sg+2n_++1) \mathcal{Z}^+_{c=-2}(\sg+2n_-)=0, \quad n_{+},n_-\in\mathbb{Z},
	\end{equation}
	which itself follows from \eqref{pmsym}.
	
	Therefore, from the last row of \eqref{Ftrcalc} it follows that
	\begin{equation}
	\tau(\sg,s|z)=\uptau^+(\sg,s|z)\uptau^-(\sg,s|z), \label{NYtaupm}
	\end{equation}
	where we use the following notation:
	\begin{defn}
		The functions $\uptau^{\pm}(\sg,s|z)$ given by the formula
		\begin{equation}
		\uptau^{\pm}(\sg,s|z)=\sum_{n\in\mathbb{Z}}s^{n/2}\mathcal{Z}^{\pm}_{c=-2}(\sg+n|z/4),   \label{taum2pm}
		\end{equation} 
		are called short $c=-2$ tau functions.
	\end{defn}
	
	On the other hand, from the penultimate row of \eqref{Ftrcalc} and due to \eqref{pmsym} we have
	\begin{equation}
	\tau(\sg,s|z)=\uptau^0(\sg,s|z)^2+\uptau^1(\sg,s|z)^2, \label{NYtau01}
	\end{equation}
	where we use another notation:
	\begin{defn}
		The functions $\uptau^i(\sg,s|z)$ given by the formula
		\begin{equation}
		\uptau^i(\sg,s|z)=\sum_{n\in\mathbb{Z}+i/2}s^{n}\mathcal{Z}_{c=-2}(\sg+2n|z/4), \quad i\in\mathbb{Z}/2\mathbb{Z}    \label{taum201}
		\end{equation}
		are called long $c=-2$ tau functions.
	\end{defn}
	
	Short and long tau functions are connected by the relations
	\begin{equation}
	\uptau^{\pm}(\sg,s|z)=(2\cos\pi\sg)^{\mp1/2}(\uptau^0(\sg,s|z) \pm i \uptau^1(\sg,s|z)). \label{lsc}
	\end{equation}
	
	These tau functions also have symmetry properties
	\begin{eqnarray}
	&\uptau^{\pm}(\sg+1,s|z)=s^{-1/2}\uptau^{\pm}(\sg,s|z), \quad &\uptau^{\pm}(-\sg,s^{-1}|z)=\uptau^{\pm}(\sg,s|z),\\
	&\uptau^i(\sg+1,s|z)=(-1)^i s^{-1/2}\uptau^{i-1}(\sg,s|z), \quad &\uptau^i(-\sg,s^{-1}|z)=\uptau^i(\sg,s|z). \label{Kievsymm2}
	\end{eqnarray}
	Note that the first column agrees with the first property from \eqref{Kievsymm} via \eqref{NYtaupm} or \eqref{NYtau01}. 
	
	The naming for "short" and "long" $c=-2$ tau functions is inspired by the length of shifts of $\sg$ in the sums \eqref{taum2pm} and \eqref{taum201} correspondingly. 
	

	\subsection{Painlev\'e III($D_8^{(1)}$) equation from $c={-}2$ tau functions}
	\label{ssec:KZdif}
	
	There are also differential Nakajima-Yoshioka blowup relations on 4d pure gauge $\mathcal{N}=2$ $SU(2)$ Nekrasov partition functions \cite{NY03},
	\cite{NY03L}, \cite{BFL13}. Relations \cite[(5.2)]{NY03L} (\cite[(6.14)]{NY03}) have the form
	\begin{equation}
	\sum_{n\in\mathbb{Z}} \mathcal{Z}(a+2\epsilon_1 n;\epsilon_1,\epsilon_2-\epsilon_1|\Lambda e^{-\frac12\epsilon_1 \alpha})
	\mathcal{Z}(a+2\epsilon_2 n;\epsilon_1-\epsilon_2,\epsilon_2|\Lambda e^{-\frac12\epsilon_2 \alpha})=O(\alpha^4),\, \textrm{for}\, \alpha\rightarrow0. \label{NY2}
	\end{equation}
	The precise coefficient of the power $\alpha^4$ was obtained in \cite{BFL13}
	\begin{equation}
	\begin{aligned}
	\sum_{n\in\mathbb{Z}} \mathcal{Z}(a+2\epsilon_1 n;\epsilon_1,\epsilon_2-\epsilon_1|\Lambda e^{-\frac12\epsilon_1 \alpha})
	\mathcal{Z}(a+2\epsilon_2 n;\epsilon_1-\epsilon_2,\epsilon_2|\Lambda e^{-\frac12\epsilon_2 \alpha})|_{\alpha^4}=\\
	=\frac{(2\alpha)^4}{4!}\left(\left(\frac{\epsilon_1+\epsilon_2}{4}\right)^4-2\Lambda^4\right)\mathcal{Z}(a;\epsilon_1,\epsilon_2|\Lambda).
	\label{NY4}
	\end{aligned}
	\end{equation}

	There are also blowup relations in the  half-integer sector ($n$ runs over $\mathbb{Z}+1/2$)\footnote{It seems that they first appear in
		$q$-deformation version \cite[(2.5)]{NY05}. Continuous case could be obtained by the limit from \eqref{qNYD12diff} for $j=1$.
	}
	\begin{equation}
	\begin{aligned}
	\sum_{n\in\mathbb{Z}+1/2} \mathcal{Z}(a+2\epsilon_1 n;\epsilon_1,\epsilon_2-\epsilon_1|\Lambda e^{-\frac12\epsilon_1 \alpha})
	\mathcal{Z}(a+2\epsilon_2 n;\epsilon_1-\epsilon_2,\epsilon_2|\Lambda e^{-\frac12\epsilon_2 \alpha})|_{\alpha^1}
	=i\alpha\Lambda\mathcal{Z}(a;\epsilon_1,\epsilon_2|\Lambda)\label{NY1}
	\end{aligned}
	\end{equation}

	Setting $\epsilon_1+\epsilon_2=0$ and moving to the CFT notations we obtain for integer shifts case
	\begin{equation}
	\begin{aligned}
	&\sum_{n\in\mathbb{Z}} D^2_{[\log z]}(\mathcal{Z}^+_{c=-2}(\sg-n|z/4),\mathcal{Z}^-_{c=-2}(\sg+n|z/4))=0,\\
	&\sum_{n\in\mathbb{Z}} D^4_{[\log z]}(\mathcal{Z}^+_{c=-2}(\sg-n|z/4),\mathcal{Z}^-_{c=-2}(\sg+n|z/4))=-2z \mathcal{Z}_{c=1}(\sg|z), \label{NYdiffIS}
	\end{aligned}
	\end{equation}
	and for half-integer shifts
	\begin{equation}
	\sum_{n\in\mathbb{Z}+1/2} D^1_{[\log z]}(\mathcal{Z}^+_{c=-2}(\sg-n|z/4),\mathcal{Z}^-_{c=-2}(\sg+n|z/4))=\frac{i}2z^{1/4}\mathcal{Z}_{c=1}(\sg|z).\label{NYdiffHIS1}
	\end{equation}
	There exists also a half-integer shift relation with Hirota derivative of order $3$ 
	\begin{equation}
	\sum_{n\in\mathbb{Z}+1/2} D^3_{[\log z]}(\mathcal{Z}^+_{c=-2}(\sg-n|z/4),\mathcal{Z}^-_{c=-2}(\sg+n|z/4))=\frac{i}2z^{1/4}\left(z\frac{d}{dz}\right)\mathcal{Z}_{c=1}(\sg|z). \label{NYdiffHIS3}
	\end{equation}
	We have not found it explicitly in the literature but it follows from the results of \cite{NY09}.\footnote{Namely it follows from the \cite[Theorem 2.6]{NY09} that left side of this equation should be a $P(z\dfrac{d}{dz},z)\mathcal{Z}_{c=1}(\sg|z)$, where $P$ is a certain polynomial with coefficients in $\mathbb{C}[\epsilon_1,\epsilon_2]$. The order of this polynomial in $z$ is bounded by the homogeneity. Then, the dependence on $z\dfrac{d}{dz}$ can determined by the action on the first nontrivial term in $z$ expansion. We are grateful to H. Nakajima for the explanation on this point.}

	
	These relations could be also rewritten in terms of the $c=-2$ tau functions just in the same manner as relations \eqref{NYtau01}
	or \eqref{NYtaupm} were obtained.
	Namely, relations \eqref{NYdiffIS} become
	\begin{eqnarray}
	D^2_{[\log z]}(\uptau^0,\uptau^0)+D^2_{[\log z]}(\uptau^1,\uptau^1)&=&D^2_{[\log z]}(\uptau^+,\uptau^-)=0, \label{NYD2diff}\\
	D^4_{[\log z]}(\uptau^0,\uptau^0)+D^4_{[\log z]}(\uptau^1,\uptau^1)&=&D^4_{[\log z]}(\uptau^+,\uptau^-)=-2z\tau. \label{NYD4diff}
	\end{eqnarray}
	To rewrite \eqref{NYdiffHIS1}, \eqref{NYdiffHIS3} in terms of tau functions we should additionaly make substitution $\sg\mapsto\sg+1/2$ and change the index of
	summation in the l.h.s. $n\mapsto n+1/2$. Then in r.h.s. we obtain B\"acklund transformed tau function $\tau_1$
	\begin{eqnarray}
	D^1_{[\log z]}(\uptau^0,\uptau^1)&=&\frac{i}2 D^1_{[\log z]}(\uptau^+,\uptau^-)=\frac{i}2z^{1/4}\tau_1,\label{NYD1diff}\\
	D^3_{[\log z]}(\uptau^0,\uptau^1)&=&\frac{i}2 D^3_{[\log z]}(\uptau^+,\uptau^-)=\frac{i}2z^{1/4}\left(z\frac{d}{dz}\right)\tau_1. \label{NYD3diff}
	\end{eqnarray}
	The blowup equations from the above express $c=1$ tau functions as a bilinear combination of $c=-2$ tau functions.
	Of course, excluding $c=1$ tau function from these relations (for instance, by the substitution of $\tau$ given by \eqref{NYtaupm}) we will obtain
	bilinear relations only on $c=-2$ tau functions.

	We have obtained many relations between $c=1$ tau function given by \eqref{Kiev} and $c=-2$ tau functions: \eqref{NYtaupm}, \eqref{NYD2diff}, \eqref{NYD4diff},
	\eqref{NYD1diff}, \eqref{NYD3diff}. Now we deduce from them Toda-like equations \eqref{eq:Toda} on tau function given by the formula \eqref{Kiev}. 
	Therefore we will obtain a new proof of Theorem \ref{Kievthm}. We will do this in two slightly different ways.
	
	\begin{prop}\label{doubleprop}
		Let $\uptau^{\pm}$ satisfy equations \eqref{NYD2diff}. Then  $\tau_0(z)=\tau(z)$ and $\tau_1(z)$ defined by the \eqref{NYtaupm}, \eqref{NYD1diff} 
		satisfy Toda-like equation \eqref{eq:Toda} for $j=0$.
	\end{prop}
	\begin{proof}
		The proof is elementary: we just substitute $\tau$ and $\tau_1$ given by \eqref{NYtaupm}, \eqref{NYD1diff}
		into the Toda-like equation and check that under \eqref{NYD2diff} it is an identity
		\begin{equation}
		D^2_{[\log z]}(\tau,\tau)=D^2_{[\log z]}(\uptau^+\uptau^-,\uptau^+\uptau^-)=-2(D^1_{[\log z]}(\uptau^+,\uptau^-))^2=-2z^{1/2}\tau_1^2,
		\end{equation}
		where we used the identity
		\begin{equation}
		D^2_{[x]}(f(x)g(x),f(x)g(x))=2 f(x) g(x) D^2_{[x]}(f(x),g(x))-2(D^1_{[x]}(f(x),g(x)))^2.
		\end{equation}
	\end{proof}
	We considered equations \eqref{NYtaupm}, \eqref{NYD1diff}, \eqref{NYD2diff} on functions depending only on $z$.
	Let us consider these equations as equations on functions depending on $\sg$ and $z$ such that
	$\tau(z)\mapsto \tau(\sg|z)$, $\tau_1(z)\mapsto s^{1/2}\tau(\sg+1/2|z)=s^{-1/2}\tau(\sg-1/2|z)$, $\uptau^{\pm}(z)\mapsto \uptau^{\pm}(\sg|z)$.
	Then previous Proposition states that the function $\tau(\sg,s|z)$ satisfy \eqref{eq:Todasg}.
	Therefore the Theorem \ref{Kievthm} follows from the Proposition \ref{doubleprop} (up to the freedom $z^K$ from Proposition \ref{prop:Toda},
	which is fixed by imposing the asymptotic behaviour \cite[Prop. 2.1.]{BS16b}).
	

	%
	
	\paragraph{Knizhnik-Zamolodchikov equation.}
	A second way to obtain Toda-like equations \eqref{eq:Toda} from the relations between $c=1$ tau function and $c=-2$ tau functions is based on first order linear differential equations on $\uptau^{\pm}$.
	We expect them to be Knizhnik-Zamolodchikov (KZ) equation on the conformal blocks of symplectic fermions.
	
	At first we could write the KZ equation on Painlev\'e VI isomonodromic  tau function. Namely, Painlev\'e VI equation is an equation on isomonodromic deformations of the rank $2$ meromorphic flat connections on $\mathbb{CP}^1$ with $4$ poles, i.e. of the linear system
	\begin{equation}
	\frac{d\Phi(t)}{dt}=A(t)\Phi(t), \quad A(t)=\frac{A_z}{t-z}+\frac{A_0}{t}+\frac{A_1}{t-1}. 
	\end{equation}
	Isomonodromic tau function is introduced by the integration of 
	closed form which is as follows
	\begin{equation}
	d\log\tau=\left(\frac{\tr A_0A_z}{z}+\frac{\tr A_1A_z}{z-1}\right) dz. \label{KZVI} \end{equation}
	The isomonodromic Painlev\'e VI tau function is a conformal block  of free complex fermions (see, for example, \cite[Sec. 2.3.]{GIL12}).
	In this case, the space of conformal blocks is $1$-dimensional and the corresponding KZ equation is just \eqref{KZVI}. 
	We can make irregular limit of \eqref{KZVI} to KZ equation on the Painlev\'e III($D_8^{(1)}$) tau function
	\begin{equation}
	z\frac{d\tau}{dz}=\zeta\tau, \label{tauzetad}
	\end{equation}
	which is just \eqref{tauzeta}. 
	
	Similarly to the $c=1$ tau function, $c=-2$ tau functions   
	are expected to be conformal blocks of symplectic fermions. 
	Space of these conformal blocks is $2$-dimensional in accordance with the two $c=-2$ tau functions we have. 
	We expect that appropriate KZ equation on the space of $c=-2$ tau functions is possible to obtain from the CFT, but in this paper we
	will just find it by hands.
	To do that, let us write an identity
	\begin{equation}
	z\frac{d}{dz}
	\begin{pmatrix}
	\uptau^0\\
	\uptau^1
	\end{pmatrix}=
	\begin{pmatrix}
	a(z) & b(z)\\
	-b(z) & a(z) 
	\end{pmatrix}
	\begin{pmatrix}
	\uptau^0\\
	\uptau^1
	\end{pmatrix},\label{KZab}
	\end{equation}
	where 
	\begin{equation}
	a(z)=\frac{\uptau^0(\uptau^0)'+\uptau^1(\uptau^1)'}{(\uptau^0)^2+(\uptau^1)^2}, \qquad
	b(z)=\frac{\uptau^1(\uptau^0)'-\uptau^0(\uptau^1)'}{(\uptau^0)^2+(\uptau_1)^2}.  \end{equation}
	Due to \eqref{NYtau01} we rewrite
	\begin{equation}
	a(z)=\frac{\tau'}{2\tau}, \qquad b(z)=\frac{D^1_{[\log z]}(\uptau^0,\uptau^1)}{\tau},\label{abdef}
	\end{equation}
	and we have from \eqref{tauzetad} $a(z)=\frac12\zeta(z)$
	Using substitution \eqref{KZab} for $(\uptau^0)',(\uptau^1)'$ 
	in the l.h.s. of \eqref{NYD2diff} twice
	\begin{equation}
	\begin{aligned}
	&\frac12D^2_{[\log z]}(\uptau^0,\uptau^0)+\frac12D^2_{[\log z]}(\uptau^1,\uptau^1)=(\uptau^0)''\uptau^0-(\uptau^0)'^2+(\uptau^1)''\uptau^1-(\uptau^1)'^2=\\
	&=(a(z)\uptau^0+b(z)\uptau^1)'\uptau^0-(a(z)\uptau^0+b(z)\uptau^1)^2+(-b(z)\uptau^0+a(z)\uptau^1)'\uptau^1-(-b(z)\uptau^0+a(z)\uptau^1)^2=\\
	&=(a'(z)-a(z)^2-b(z)^2)((\uptau^0)^2+(\uptau^1)^2)+
	(a(z)(\uptau^0)'+b(z)(\uptau^1)')\uptau^0+
	(-b(z)(\uptau^0)'+a(z)(\uptau^1)')\uptau^1=\\
	&=(a'(z)-a(z)^2-b(z)^2)((\uptau^0)^2+(\uptau^1)^2)+
	(a(z)(a(z)\uptau^0+b(z)\uptau^1)+b(z)(-b(z)\uptau^0+a(z)\uptau^1))\uptau^0+\\
	&+(-b(z)(a(z)\uptau^0+b(z)\uptau^1)+a(z)(-b(z)\uptau^0+a(z)\uptau^1))\uptau^1=(a'(z)-2b(z)^2)((\uptau^0)^2+(\uptau^1)^2),
	\end{aligned}
	\label{D2KZcalc}
	\end{equation}
	which should be equal to zero according to \eqref{NYD2diff}, we obtain that 
	$b(z)=\frac12\sqrt{\zeta'}$
	
	So \eqref{KZab} is system of first order linear differential equations on $\uptau_0,\uptau_1$ with coefficients expressed in terms of 
	Painlev\'e III($D_8^{(1)}$) function $\zeta(z)$.
	
	It follows from \eqref{lsc} that short $c=-2$ tau functions are "eigenfunctions" for $z\frac{d}{dz}$ with "eigenvalues" $\frac12(\zeta\mp i\sqrt{\zeta'})$
	\begin{equation}
	z\frac{d}{dz}\uptau^{\pm}=\frac12(\zeta\mp i\sqrt{\zeta'})\uptau^{\pm}. \label{KZ}
	\end{equation}
	
	Above we did not use that $\zeta$ and corresponding $\tau$
	are solutions of Painlev\'e III($D_8^{(1)}$) equation.
	Now we will deduce another proof of the Theorem \ref{Kievthm} from the     KZ equation \eqref{KZ} and bilinear equation \eqref{NYD4diff}.
	\begin{prop}
		Let functions $\uptau^{\pm}$ satisfy equations \eqref{KZ}
		and bilinear equation \eqref{NYD4diff}.
		Then there exist such complex number $K$ that $\zeta-K$ satisfy \eqref{zeta3}.
	\end{prop}
	
	\begin{proof}
		Let us use the substitution \eqref{KZ} for $(\uptau^+)',(\uptau^-)'$ in the l.h.s. of \eqref{NYD4diff} four times analogously to the calculation \eqref{D2KZcalc}.
		Then from \eqref{NYD4diff} and \eqref{NYtaupm} we obtain equation on $\zeta(z)$
		\begin{equation}
		-2(\zeta')^3+(\zeta'')^2-\zeta'\zeta'''+2z\zeta'=0. \label{eq:zetac}
		\end{equation}
		
		This equation is almost equivalent to the \eqref{zeta3} (cf. \cite[(2.26)]{BS16b}).
		Indeed, following proof of \cite[Prop. 2.3]{BS16b} (which is Proposition \ref{prop:Toda}) we rewrite equation \eqref{zeta3} in form
		$f(z)=0$, where
		\begin{equation}
		f(z)=\frac{1}{z^2}((\zeta''-\zeta')^2-4\zeta'^2(\zeta-\zeta')+4z\zeta').
		\end{equation}
		
		Differentiating the expression for $f(z)$, we have
		\begin{equation}
		\frac{z^2}{2(\zeta''-\zeta')}f'=\zeta'''-2\zeta''+\zeta'+6\zeta'^2-4\zeta\zeta'+2z.
		\end{equation}
		Then equation \eqref{eq:zetac} can be rewritten as
		\begin{equation}
		z^2f=\frac{z^2\zeta'}{2(\zeta''-\zeta')}f' \Leftrightarrow 2\ddot\zeta f=\dot\zeta \dot f \Leftrightarrow f=4K\dot\zeta^2.
		\end{equation}
		Recall that dot means differentiation with respect to $z$.
		This means that $\zeta-K$ satisfies \eqref{zeta3}. 
	\end{proof}
	
	\begin{Remark}
		This $K$ is the same as in Proposition \ref{prop:Toda}.
	\end{Remark}
	
	\begin{Remark}
		We deduced Painlev\'e III($D_8^{(1)}$) equation from Nakajima-Yoshioka blowup relations without using the half-integer sector
		\eqref{NYD1diff} of them. Moreover, it appears that \eqref{NYD1diff} follows from integer-sector relations using the Painlev\'e III($D_8^{(1)}$) equation.
		Indeed, we have from the first Toda-like equation \eqref{eq:Toda}
		that $\zeta'=-z^{1/2}\frac{\tau_1^2}{\tau^2} $ and from the  second relation of \eqref{abdef} it follows that   $\zeta'=4\left(\frac{D^1_{[\log z]}(\uptau^0,\uptau^1)}{\tau}\right)^2$.
	\end{Remark}
	
	\paragraph{Algebraic solution and $c=-2$ tau functions.}
	Concluding the section let us find $\uptau^{\pm}$ corresponding to the algebraic solution, see equation \eqref{eq:algebraic} and discussion below it.
	Substituting into \eqref{NYD1diff} $\tau_1=\tau$ and using \eqref{NYtaupm} we obtain
	\begin{equation}
	D^1_{[\log z]}(\uptau^+,\uptau^-)=z^{1/4}\uptau^+\uptau^-,\quad\textrm{i.e.}\quad (\log \uptau^+-\log \uptau^-)'=z^{1/4}.        
	\end{equation}
	Integrating this we obtain $\uptau^+=e^{4 z^{1/4}}\uptau^-$. 
	
	Two algebraic solutions $\tau=\uptau^+\uptau^-=z^{1/16}e^{\mp 4z^{1/2}}$ just correspond to the two choices of branch of $z^{1/2}$.
	We already choose some branch of $z^{1/4}$ when use equation \eqref{NYD1diff}. This leads to choice of the branch of $z^{1/2}=(z^{1/4})^2$ in the above calculation, i.e. to the choice of one of the algebraic solutions.
	The right and wrong choices give us
	\begin{equation}
	\uptau^{+}=z^{1/32}e^{2z^{1/4}\mp2z^{1/2}}, \quad
	\uptau^{-}=z^{1/32}e^{-2z^{1/4}\mp2z^{1/2}}.
	\end{equation}
	One could easily check that functions $\uptau^{+}, \uptau^{-}$ given by the previous formula with the sign "-" satisfy 	\eqref{NYD2diff}, but with the sign “+” do not satisfy.
	So the correct answer for $\uptau^{\pm}$ is given  by previous formula with the sign "-".

	Analytic continuation around $z=0$ give us $\uptau^{\pm}$, which product is $\tau=\tau_1=z^{1/16}e^{+4\sqrt{z}}$
	\begin{equation}
	\uptau^{\pm}=z^{1/32}e^{\pm 2i z^{1/4}+2z^{1/2}}.
	\end{equation}
	If we traverse the cycle around  $z=0$ twice then we obtain the initial answers with permuted $\uptau^+$ and $\uptau^-$.

	\section{$q$-deformed $c={-}2$ tau functions}
	\label{sec:qtau}
	\subsection{$q$-deformed $c={-}2$ tau functions and $q$-Painlev\'e $A_7^{(1)'}$ equation}
	\label{ssec:qtaum2}
	Nakajima-Yoshioka blowup relations on 5d partition functions \eqref{5d_def} are given by
	\cite[(2.5)-(2.7)]{NY05}. We rewrite these equations as equations on full partition functions
	\begin{align}
	(q_1^{-1/4}q_2^{-1/4}\Lambda)^{j}\mathcal{Z}(u;q_1,q_2|\Lambda&)=\sum_{n\in\mathbb{Z}+j/2}\mathcal{Z}(uq_1^{2n};q_1,q_2q_1^{-1}|q_1^{-\frac14}\Lambda) \mathcal{Z}(uq_2^{2n};q_1q_2^{-1},q_2|q_2^{-\frac14}\Lambda),\label{qNY1}\\
	(1-j)\mathcal{Z}(u;q_1,q_2|\Lambda)&=\sum_{n\in\mathbb{Z}+j/2}\mathcal{Z}(uq_1^{2n};q_1,q_2q_1^{-1}|\Lambda) \mathcal{Z}(uq_2^{2n};q_1q_2^{-1},q_2|\Lambda), \label{qNY2}\\
	(-q_1^{1/4}q_2^{1/4}\Lambda)^{j}\mathcal{Z}(u;q_1,q_2|\Lambda)&=\sum_{n\in\mathbb{Z}+j/2}\mathcal{Z}(uq_1^{2n};q_1,q_2q_1^{-1}|q_1^{\frac14}\Lambda) \mathcal{Z}(uq_2^{2n};q_1q_2^{-1},q_2|q_2^{\frac14}\Lambda), \label{qNY3}
	\end{align}
	where $j=0,1$.
	Moving to the case $q_2=q_1^{-1}=q$ (i.e. $\epsilon_1+\epsilon_2=0$) we obtain that \eqref{qNY1} and \eqref{qNY3} become identical
	and both have the form
	\begin{equation}
	z^{j/4}\mathcal{Z}(u;q^{-1},q|z)=\sum_{n\in\mathbb{Z}+j/2}\mathcal{Z}(uq^{-2n};q^{-1},q^2|qz) \mathcal{Z}(uq^{2n};q,q^{-2}|q^{-1}z), \quad j=0,1. \label{qNYD12diff}\\
	\end{equation}
	The relation \eqref{qNY2} in the case $q_2=q_1^{-1}=q$ for $j=1$ becomes trivial and for $j=0$ it reads
	\begin{equation}
	\mathcal{Z}(u;q^{-1},q|z)=\sum_{n\in\mathbb{Z}}\mathcal{Z}(uq^{-2n};q^{-1},q^2|z) \mathcal{Z}(uq^{2n};q,q^{-2}|z).
	\end{equation}
	
	Analogously to the continuous case we obtain from the last relation
	\begin{equation}
	\tau(u,s|z)=\uptau^+(u,s|z)\uptau^-(u,s|z), \label{qNYtaupm}
	\end{equation}
	where we use the following notation:
	\begin{defn}
		The functions $\uptau^{\pm}(u,s|z)$ given by the formula
		\begin{equation}
		\uptau^{\pm}(u,s|z)=\sum_{n\in\mathbb{Z}}s^{n/2}\mathcal{Z}(uq^{2n};q^{\mp1},q^{\pm2}|z)  \label{qtaum2pm}
		\end{equation}
		are called short $q$-deformed $c=-2$ tau functions.
	\end{defn}
	
	From the relation \eqref{qNYD12diff} for $j=0,1$ we obtain $q$-difference equations on $\uptau^{\pm}$
	\begin{eqnarray}
	j=0: \quad \overline{\uptau^+}\underline{\uptau^-}+\underline{\uptau^+}\overline{\uptau^-}&=&2\tau, \label{qNYD2diff}\\
	j=1: \quad \overline{\uptau^+}\underline{\uptau^-}-\underline{\uptau^+}\overline{\uptau^-}&=&-2z^{1/4}\tau_1.\label{qNYD1diff}
	\end{eqnarray}
	Excluding $\tau$ from \eqref{qNYD2diff} and \eqref{qNYtaupm} we obtain
	equation only on $c=-2$ tau functions
	\begin{equation}
	\overline{\uptau^+}\underline{\uptau^-}+\underline{\uptau^+}\overline{\uptau^-}=2\uptau^+\uptau^-. \label{qNYD2diffp} 
	\end{equation}
	
	Then we have an analog of the Proposition \ref{doubleprop}
	\begin{prop}
		Let $\uptau^{\pm}$ satisfy equations \eqref{qNYD2diffp}. Then  $\tau_0(z)=\tau(z)$ and $\tau_1(z)$ defined by  \eqref{qNYtaupm}, \eqref{qNYD1diff} correspondingly 
		satisfy Toda-like equation \eqref{eq:qToda} for $j=0$.\label{qdoubleprop}
	\end{prop}
	\begin{proof}
		The proof is even more elementary than in the continuous case. Namely
		\begin{equation}
		\begin{aligned}
		\overline{\tau}\underline{\tau} =\overline{\uptau^+} \overline{\uptau^-} \underline{\uptau^+} \underline{\uptau^-}=
		\frac14(\overline{\uptau^+} \underline{\uptau^-}+\underline{\uptau^+} \overline{\uptau^-})^2-
		\frac14(\overline{\uptau^+} \underline{\uptau^-}-\underline{\uptau^+} \overline{\uptau^-})^2=\tau^2-z^{1/2}\tau_1^2.
		\end{aligned}
		\end{equation}
		\end{proof}
	Toda-like equation \eqref{eq:qTodasg} on $\tau(u|z)$ follows from this proposition as in the continuous case.
	Thus this proposition gives us automatical proof of the Theorem \ref{qtauconj}, deducing it from the Nakajima-Yoshioka
	blowup equations.
	
	We obtained $q$-relations \eqref{qNYtaupm}, \eqref{qNYD2diff}, \eqref{qNYD1diff}, their continuous analogs
	are \eqref{NYtaupm}, \eqref{NYD2diff}, \eqref{NYD1diff} correspondingly. 
	The analog of \eqref{NYD4diff} could be obtained from the previously listed equations
	\begin{multline}
	(\overline{\overline{\uptau^+}} \underline{\underline{\uptau^-}}+\underline{\underline{\uptau^+}} \overline{\overline{\uptau^-}})\uptau^+\uptau^-=
	(\overline{\tau}-(qz)^{1/4}\overline{\tau_1})(\underline{\tau}-(q^{-1}z)^{1/4}\underline{\tau_1})\\+
	(\overline{\tau}+(qz)^{1/4}\overline{\tau_1})(\underline{\tau}+(q^{-1}z)^{1/4}\underline{\tau_1})
	=2\overline{\tau}\underline{\tau}+2 z^{1/2}\overline{\tau_1}\underline{\tau_1}=2(1-z)\tau^2,	    
	\end{multline}
	where we used \eqref{qNYD2diff}, \eqref{qNYD1diff}     and then \eqref{eq:qToda}.
	Therefore from \eqref{qNYtaupm} we obtain
	\begin{equation}
	\overline{\overline{\uptau^+}} \underline{\underline{\uptau^-}}+\underline{\underline{\uptau^+}} \overline{\overline{\uptau^-}}=2(1-z)\tau. \label{qNYD4diff}
	\end{equation}
	
	\paragraph{Algebraic solution and $q$-deformed $c=-2$ tau functions.}
	Let us now find $\uptau^{\pm}$ corresponding to the algebraic $c=1$ tau function of Painlev\'e $A_7^{(1)'}$ equation given by
	$\tau=\tau_1=z^{1/16}(\pm q^{1/2}z^{1/2};q^{1/2},q^{1/2})_{\infty}$.
	Substituting $\tau=\tau_1$, we have from \eqref{qNYD2diff}, \eqref{qNYD1diff}
	\begin{equation}
	\overline{\uptau^+} \underline{\uptau^-}=(1-z^{1/4})\tau. 
	\end{equation}
	Then, dividing the both sides by $\underline{\tau}$ and using 
	\eqref{qNYtaupm} we obtain
	\begin{equation}
	\frac{\overline{\uptau^+}}{\underline{\uptau^+}}=(1-z^{1/4})\frac{\tau}{\underline{\tau}}.    
	\end{equation}
	
	As in the continuous case, we have to choose the branch of $z^{1/2}$ in  $(\pm q^{1/2}z^{1/2};q^{1/2},q^{1/2})_{\infty}$ which agrees
	with \eqref{qNYD1diff}. Now we will make calculations with both choices and finally find that one of them leads to contradiction.
	Substituting $\tau=z^{1/16}(\pm q^{1/2}z^{1/2};q^{1/2},q^{1/2})_{\infty}$ we obtain
	\begin{equation}
	\frac{\overline{\uptau^+}}{\underline{\uptau^+}}=q^{1/16}\frac{1-z^{1/4}}{(\pm z^{1/2};q^{1/2})_{\infty}},
	\end{equation}
	where we used \eqref{qshift} and then \eqref{permult}.
	We have that
	\begin{equation}
	\uptau^+=
	z^{1/32} \frac{(\pm q^{1/2}z^{1/2};q^{1/2},q)_{\infty}}{(q^{1/4}z^{1/4};q^{1/2})_{\infty}}
	\end{equation}
	is the only solution (up to $z$-constant factor) of the previous equation which is a power series in $z$. 
	Analogously from $\overline{\uptau^-} \underline{\uptau^+}=(1+z^{1/4})\tau$ we obtain
	\begin{equation}
	\uptau^-=z^{1/32} \frac{(\pm q^{1/2}z^{1/2};q^{1/2},q)_{\infty}}{(-q^{1/4}z^{1/4};q^{1/2})_{\infty}}. 
	\end{equation}
	Let us now check \eqref{qNYtaupm}. For the branch corresponding to the sign "+" we obtain
	\begin{equation}
	\uptau^+\uptau^-=z^{1/16} \frac{(q^{1/2}z^{1/2};q^{1/2},q)^2_{\infty}}{(q^{1/2}z^{1/2};q)_{\infty}}=
	z^{1/16}(q^{1/2}z^{1/2};q^{1/2},q^{1/2})_{\infty},
	\end{equation}
	where we used \eqref{sqmult}, \eqref{qshift} and \eqref{permult} successively. 
	Of course, there is no such relation if we choose sign "-", so the latter choice is wrong and the previous choice give the correct answer.
	
	As in the continuous case, answer for the other branch of $z^{1/2}$ could be obtained by the analytic continuation around $z=0$.
	Going twice around $z=0$ permutes $\uptau^+$ and $\uptau^-$.

	\subsection{Chern-Simons generalization}\label{ssec:CS}
	The work \cite{BGM18} considered a generalization of the Toda-like equations \eqref{eq:qToda}. This generalization
	depends on two integer parameters $N\in\mathbb{N}, 0\leq m \leq N$ and has the form
	\begin{equation}
	\tau_{m;j}(qz)\tau_{m;j}(q^{-1}z)=\tau_{m;j}(z)^2-z^{1/N}\tau_{m;j+1}(q^{m/N}z)\tau_{m;j-1}(q^{-m/N}z), \, j\in \mathbb{Z}/N\mathbb{Z}. 
	\label{eq:qTodaCSg}
	\end{equation}
	Clearly, the original equations \eqref{eq:qToda} correspond to the case $N=2,m=0$.
	In this paper we consider only the cases $N=2, m=0,1,2$, so we consider equations
	\begin{equation}
	\tau_{m;j}(qz)\tau_{m;j}(q^{-1}z)=\tau_{m;j}(z)^2-z^{1/2}\tau_{m;j+1}(q^{m/2}z)\tau_{m;j-1}(q^{-m/2}z), \, j\in \mathbb{Z}/2\mathbb{Z},\, 0\leq m \leq 2, 
	\label{eq:qTodaCS}
	\end{equation}
	and their solutions.
	
	The work \cite{BGM18} also proposed solutions of the equations \eqref{eq:qTodaCSg} for arbitrary $N$ and $0 \leq m \leq N$. For the case $N=2$
	they are given in the form \eqref{Kiev} with the modification of the 5d partition function by the Chern-Simons term. 
	This modification is as follows \cite{T04}, \cite{GNY06}: we multiply each summand of $\mathcal{Z}_{inst}$ in \eqref{5d_def} by the
	multiplier $\prod_{i=1}^2 (q_1q_2)^{-\frac{m}2 |\lmb^{(i)}|} \mathsf{T}^m_{\lmb^{(i)}}(u_i;q_1,q_2)$ where 
	\begin{equation}
	\mathsf{T}_{\lmb}(u;q_1,q_2)=\prod_{(i,j)\in\lmb} u^{-1}q_1^{1-i}q_2^{1-j}. \label{CSterm}
	\end{equation}
	The factors $\mathcal{Z}_{cl}$ and $\mathcal{Z}_{1-loop}$ remain unchanged under this modification.
	The index $m$ is the Chern-Simons level. We will denote Chern-Simons modified full 5d partition functions by $\mathcal{Z}_m$. 
	
	For the function $\mathcal{Z}_m$, the symmetries $q_1\leftrightarrow q_2$ and $u\mapsto u^{-1}$ from \eqref{qZsymm} are obviously satisfied for arbitrary $m$.
	In the case $q_1q_2=1$ the symmetry $q_1, q_2\mapsto q^{-1}_1, q^{-1}_2$ is equivalent to the symmetry $q_1\leftrightarrow q_2$.
	
	For general $q_1, q_2$, the situation with $q_1, q_2\mapsto q^{-1}_1, q^{-1}_2$ symmetry is much more subtle. In the case $m=0$ the proof of
	such symmetry is based on the power series term by term coincidence. But for $m=1,2$ this method does not work.
	For $m=1$, however, one has $\mathcal{Z}_{1,inst}(u;q_1,q_2|\Lambda)=\mathcal{Z}_{1,inst}(u;q_1^{-1},q_2^{-1}|\Lambda)$. 
	The proof for the $q_1=q^{-1},q_2=q^2$ case is given in \cite[Prop. 1.38]{GNY06}.
	For $m=2$ it is satisfied with some elementary multiplier and in the case $q_1=q^{-1},q_2=q^2$ the answer is given below.
	
	Instead of \eqref{eq:qTodaCS}, we will consider a single equation on $\tau_m(u|z)$
	\begin{equation}
	\tau_{m}(u|qz)\tau_{m}(u|q^{-1}z)=\tau_{m}^2(u|z)-z^{1/2}\tau_{m}(uq|q^{m/2}z)\tau_{m}(uq^{-1}|q^{-m/2}z), \label{eq:qTodaCSsg}
	\end{equation}
	just in the same way as before we proceed from \eqref{eq:qToda} to \eqref{eq:qTodasg}.
	
	There is analog of the blowup relations on Nekrasov functions for the Chern-Simons modified case. They were proposed in
	\cite[(1.37)]{GNY06} and proven as Theorem 2.11 in \cite{NY09}. In our notations they read ($0\leq m\leq N$)
	\begin{align}
	\mathcal{Z}_m(u;q_1,q_2|\Lambda)&=\sum_{n\in\mathbb{Z}}\mathcal{Z}_m(uq_1^{2n};q_1,q_2q_1^{-1}|q_1^{-\frac14-\frac{m}8}\Lambda) \mathcal{Z}_m(uq_2^{2n};q_1q_2^{-1},q_2|q_2^{-\frac14-\frac{m}8}\Lambda), \label{qNYCS1} \\
	\mathcal{Z}_m(u;q_1,q_2|\Lambda)&=\sum_{n\in\mathbb{Z}}\mathcal{Z}_m(uq_1^{2n};q_1,q_2q_1^{-1}|q_1^{-\frac{m}{8}}\Lambda) \mathcal{Z}_m(uq_2^{2n};q_1q_2^{-1},q_2|q_2^{-\frac{m}8}\Lambda), \label{qNYCS2}\\
	\mathcal{Z}_m(u;q_1,q_2|\Lambda)&=\sum_{n\in\mathbb{Z}}\mathcal{Z}_m(uq_1^{2n};q_1,q_2q_1^{-1}|q_1^{\frac14-\frac{m}{8}}\Lambda) \mathcal{Z}_m(uq_2^{2n};q_1q_2^{-1},q_2|q_2^{\frac14-\frac{m}{8}}\Lambda) \label{qNYCS3}. 
	\end{align}
	These equations are analogs of \eqref{qNY1}, \eqref{qNY2}, \eqref{qNY3} for $j=0$. We will comment on the $j=1$ sector 
	when consider the case $m=1$ where these relations are necessary.
	
	Below we consider cases $m=1$ and $m=2$ separately.
	
	\paragraph{Case $m=2$.}
	We obtain that equation \eqref{eq:qTodaCS} for $m=2$ is equivalent to $m=0$ equation \eqref{eq:qToda}.
	\begin{prop}
		The functions $\tau_{2;j}$, $j\in\mathbb{Z}/2\mathbb{Z}$ satisfy \eqref{eq:qTodaCS} iff
		the functions $\tau_j=(qz;q,q)_{\infty}\tau_{2;j}$ satisfy \eqref{eq:qToda}.
	\end{prop}
	This was noticed in \cite{BGM18}.
	\begin{proof}
		Assume that the functions $\tau_{2;j}$, $j\in\mathbb{Z}/2\mathbb{Z}$ satisfy \eqref{eq:qTodaCS}. Then, combining equations \eqref{eq:qTodaCS} for $j$ and $j+1$ we obtain
		\begin{equation}
		\overline{\tau_{2;j}}\underline{\tau_{2;j}}=\tau_{2;j}^2-z^{1/2}\overline{\tau_{2;j+1}}\underline{\tau_{2;j+1}}
		=\tau_{2;j}^2-z^{1/2}\tau_{2;j+1}\tau_{2;j-1}+z\overline{\tau_{2;j}}\underline{\tau_{2;j}}. 
		\end{equation}
		Therefore
		\begin{equation}
		(1-z)\overline{\tau_{2;j}}\underline{\tau_{2;j}}=\tau_{2;j}^2-z^{1/2}\tau_{2;j+1}\tau_{2;j-1},\quad j\in \mathbb{Z}/2\mathbb{Z}, \label{eq:qTodaCS2r}
		\end{equation}
		so due to \eqref{qshift} we obtain that $\tau_j=(qz;q,q)_{\infty}\tau_{2;j}$ satisfy \eqref{eq:qToda}.
		
		Conversely, assume that the functions $\tau_j=(qz;q,q)_{\infty}\tau_{2;j}$ satisfy \eqref{eq:qToda}. Taking equations \eqref{eq:qToda}, defining $\tau_{2;j}=(qz;q,q)^{-1}_{\infty}\tau_j$ we obtain \eqref{eq:qTodaCS2r}.
		Combining these equations for $j$ and $j+1$ we obtain
		\begin{equation}
		(1-z)\overline{\tau_{2;j}}\underline{\tau_{2;j}}=\tau_{2;j}^2-z^{1/2}\tau_{2;j+1}\tau_{2;j+1}=
		(1-z)\tau_{2;j}^2-z^{1/2}(1-z)\overline{\tau_{2;j+1}}\underline{\tau_{2;j+1}}.
		\end{equation}
	\end{proof}
	
	Analogous equivalence holds on the level of solutions of the form  \eqref{Kiev}. 
	Moreover it holds on the level of Nekrasov functions
	\begin{prop}\label{20equivprop}
		Nekrasov functions $\mathcal{Z}_m$ satisfy
		\begin{eqnarray}
		\mathcal{Z}_2(u;q^{-2},q|z)=(z;q^{-2},q)_{\infty}\mathcal{Z}_0(u;q^{-2},q|z), \label{20equiv1} \\
		\mathcal{Z}_2(u;q^{-1},q^2|z)=(z;q^{-1},q^2)_{\infty}\mathcal{Z}_0(u;q^{-1},q^2|z), \label{20equiv2} \\
		\mathcal{Z}_2(u;q^{-1},q|z)=(z;q^{-1},q)_{\infty}\mathcal{Z}_0(u;q^{-1},q|z). \label{20equiv12}
		\end{eqnarray}
	\end{prop}
	Relations \eqref{20equiv1}, \eqref{20equiv2} clearly lead to similar relations between $c=-2$ tau functions
	$\uptau_2^{\pm}$ and $\uptau^{\pm}_0$. 
	
	In terms of topological strings this relation means a  relation between
	the geometry of local $\mathbb{F}_0=\mathbb{P}^1\times\mathbb{P}^1$ and local Hirzebruch surface $\mathbb{F}_2$. For example, the relation between Gopakumar-Vafa invariants of these manifolds is given in e.g. \cite[eq. (94)]{IKP02}. We have not found the statement of Proposition \ref{20equivprop} in the literature (except of \eqref{20equiv12}
	which appeared in \cite{BGM18}), but it maybe known. We prove it below using the blowup relations \eqref{qNYCS1}, \eqref{qNYCS2}, \eqref{qNYCS3}.
	
	From this proposition using \eqref{qZsymm}, \eqref{qloopassym}, \eqref{qclassym} we have property
	\begin{equation}
	\mathcal{Z}_2(u;q^{-1},q^2|z)=q^{-\frac{\log^2u}{4\log^2q}}\frac{1}{(zq;q^2)_{\infty}\theta(uq;q^2)}\mathcal{Z}_2(u;q^{-2},q|z),  
	\end{equation}
	this is the form of $q_1,q_2\leftrightarrow q_1^{-1},q_2^{-1}$ symmetry for $m=2$ in the case $q_1=q^{-1},q_2=q^2$.
	
	Consider relations \eqref{qNYCS1}, \eqref{qNYCS2}, \eqref{qNYCS3} for $m=2$ and exclude from them $\mathcal{Z}_m(u;q_1,q_2|\Lambda)$
	\begin{equation}
	\begin{aligned}
	&\sum_{n\in\mathbb{Z}}\mathcal{Z}_2(uq_1^{2n};q_1,q_2q_1^{-1}|q_1^{-\frac12}\Lambda) \mathcal{Z}_2(uq_2^{2n};q_1q_2^{-1},q_2|q_2^{-\frac12}\Lambda)=\\
	=&\sum_{n\in\mathbb{Z}}\mathcal{Z}_2(uq_1^{2n};q_1,q_2q_1^{-1}|q_1^{-\frac14}\Lambda) \mathcal{Z}_2(uq_2^{2n};q_1q_2^{-1},q_2|q_2^{-\frac14}\Lambda)=\\
	=&\sum_{n\in\mathbb{Z}}\mathcal{Z}_2(uq_1^{2n};q_1,q_2q_1^{-1}|\Lambda) \mathcal{Z}_2(uq_2^{2n};q_1q_2^{-1},q_2|\Lambda)\label{qCSdetermrel}.
	\end{aligned}
	\end{equation}
	
	These are equations on $\mathcal{Z}_{2,inst}$, i.e. on coefficients $c_k^{(1)}$, $c_k^{(2)}, k\in\mathbb{Z}_{\geq 0}$ of the power series 
	$\mathcal{Z}_{2,inst}(u;q_1,q_2q_1^{-1}|\Lambda)=\sum_{k=0}^{+\infty} c_k^{(1)}(u;q_1,q_2) \Lambda^{4k}$ and 
	$\mathcal{Z}_{2,inst}(u;q_1q_2^{-1},q_2|\Lambda)=\sum_{k=0}^{+\infty} c_k^{(2)}(u;q_1,q_2) \Lambda^{4k}$.
	This is because relations \eqref{qCSdetermrel} split into the relations corresponding to powers $\Lambda^{4k}, k\in\mathbb{Z}_{\geq0}$ 
	(up to the power $\Lambda^{-\frac{\log^2u}{\log q_1\log q_2}}$ from $\mathcal{Z}_{cl}$). 
	
	\begin{lemma}
		Relations \eqref{qCSdetermrel} recursively determine the coefficients $c_k^{(1)}, c_k^{(2)}, k\in\mathbb{N}$
		starting from  $c_0^{(1)}=c_0^{(2)}=1$.\label{determlemma}
	\end{lemma}
	
	\begin{proof}
		Let us take the coefficient of the power $\Lambda^{4k}$ in the relation \eqref{qCSdetermrel}, then the coefficients $c_k^{(1)}$, $c_k^{(2)}$ 
		will appear only for $n=0$. Other coefficients in these relations are known due to the induction supposition,
		therefore we have two linear equations for two unknown variables $c^{(1)}_k$, $c_k^{(2)}$.
		The fundamental matrix of these equations is
		\begin{equation}
		\begin{pmatrix}
		q_1^{-k}-1 & q_2^{-k}-1\\
		q_1^{-2k}-1 & q_2^{-2k}-1 
		\end{pmatrix}, 
		\end{equation}
		and its determinant is equal to $(q_1^{-k}-1)(q_2^{-k}-1)(q_2^{-k}-q_1^{-k})$
		which is non-zero iff none of $q_1,q_2, q_1/q_2$ is a root of unity. 
	\end{proof}
	In the sector $|q_1|\lessgtr1, |q_2|\gtrless1$, these special cases are not realized.

	\begin{proof}[Proof of the Proposition \ref{20equivprop}]
		In the Subsection \ref{ssec:qtaum2} we have proved relation \eqref{qNYD4diff}.  
		Excluding $c=1$ tau function from the relations \eqref{qNYtaupm}, \eqref{qNYD2diff}, \eqref{qNYD4diff} 
		we obtain relations on $c=-2$ tau functions which are equivalent to the relations on the Nekrasov functions
		\begin{equation}
		\begin{aligned}
		&(1-z)^{-1}\sum_{n\in\mathbb{Z}}\mathcal{Z}_0(uq^{-2n};q^{-1},q^2|q^2z) \mathcal{Z}_0(uq^{2n};q^{-2},q|q^{-2}z)=\\
		=&\sum_{n\in\mathbb{Z}}\mathcal{Z}_0(uq^{-2n};q^{-1},q^2|qz) \mathcal{Z}_0(uq^{2n};q^{-2},q|q^{-1}z)=\\
		=&\sum_{n\in\mathbb{Z}}\mathcal{Z}_0(uq^{-2n};q^{-1},q^2|z) \mathcal{Z}_0(uq^{2n};q^{-2},q|z). 
		\end{aligned}
		\end{equation}
		Let us replace $\mathcal{Z}_0$ with $\td{\mathcal{Z}}_2$ formally defined by \eqref{20equiv1}, \eqref{20equiv2}. Then, using formulas 
		\begin{equation}
		\frac{(q^{-1}z;q,q^{-2})_{\infty}(qz;q^{-1},q^2)_{\infty}}{(z;q,q^{-2})_{\infty}(z;q^{-1},q^2)_{\infty}}=1, \quad
		\frac{(q^{-2}z;q,q^{-2})_{\infty}(q^2z;q^{-1},q^2)_{\infty}}{(z;q,q^{-2})_{\infty}(z;q^{-1},q^2)_{\infty}}=\frac{1}{1-z},
		\end{equation}
		we obtain that $\td{\mathcal{Z}}_2$ satisfies \eqref{qCSdetermrel} with $q_1q_2=1$.
		Therefore due to the Lemma \ref{determlemma} $\td{\mathcal{Z}}_2=\mathcal{Z}_2$  (for general $q$).
		Hence relations \eqref{20equiv1}, \eqref{20equiv2} are proved.
		Relation \eqref{20equiv12} follows from the  equations \eqref{qNYCS3} and \eqref{qNY2} for $j=0$ on $\mathcal{Z}_2$ and $\mathcal{Z}_0$ in the case $q_1q_2=1$ via the the simple identity
		\begin{equation}
		(z;q^{-2},q)_{\infty}(z;q^{-1},q^2)_{\infty}=(z;q^{-1},q)_{\infty},
		\end{equation}
		which is due to \eqref{qtrans} and \eqref{permult}.
	\end{proof}
	
	\paragraph{Case $m=1$.}
	As it was observed in \cite{BGM18} in this case Toda equations \eqref{eq:qTodaCS}
	are equivalent to the Painlev\'e $A_7^{(1)}$ equation
	\begin{equation}
	\overline{\overline{\tau}}\underline{\underline{\tau}}=\tau^2-z^{1/2}\overline{\tau}\underline{\tau}.
	\end{equation}
	Note that this equation is not equivalent to the Painlev\'e $A_7^{(1)'}$ equation \eqref{eq:qToda}.
	
	The following theorem is an $m=1$ analog of Theorem \ref{qKievthm}. 
	\begin{thm}
		The function $\tau_{1;0}$ given by the formula \eqref{Kiev} with $\mathcal{Z}=\mathcal{Z}_1(u;q^{-1},q|z)$ satisfies Toda-like equation \eqref{eq:qTodaCSsg} for $m=1$.
	\end{thm}
	\begin{proof}The substitution of \eqref{Kiev} into \eqref{eq:qTodaCSsg} leads to a bilinear
		relation on function $\mathcal{Z}_1$. As before we want to deduce it from the
		Nakajima-Yoshioka blowup relations. To do that we need not only relations \eqref{qNYCS1}, \eqref{qNYCS2}, \eqref{qNYCS3}
		for the integer sector but also relations in the half-integer sector as it was for the proof of the Proposition \ref{qdoubleprop}.
		
		Let us consider such relations for general $q_1, q_2$.
		There is analog \cite[(1.43)]{GNY06}(proved by the Theorem 2.11 in \cite{NY09}) of relation \eqref{qNY2} for $j=1$
		\begin{equation}
		\sum_{n\in\mathbb{Z}+1/2}\mathcal{Z}_m(uq_1^{2n};q_1,q_2q_1^{-1}|\Lambda) \mathcal{Z}_m(uq_2^{2n};q_1q_2^{-1},q_2|\Lambda)=0,
		\end{equation}
		which becomes trivial in the case $q_1q_2=1$. We will use another relations which have the form 
		\begin{equation}
		\begin{aligned}
		q_1^{-1}q_2^{-1}\Lambda\mathcal{Z}_m(u;q_1,q_2|\Lambda)&=\sum_{n\in\mathbb{Z}+1/2}\mathcal{Z}_m(uq_1^{2n};q_1,q_2q_1^{-1}|q_1^{-\frac14}\Lambda) \mathcal{Z}_m(uq_2^{2n};q_1q_2^{-1},q_2|q_2^{-\frac14}\Lambda), \\
		-q_1q_2\Lambda\mathcal{Z}_m(u;q_1,q_2|\Lambda)&=\sum_{n\in\mathbb{Z}+1/2}\mathcal{Z}_m(uq_1^{2n};q_1,q_2q_1^{-1}|q_1^{\frac14}\Lambda) \mathcal{Z}_m(uq_2^{2n};q_1q_2^{-1},q_2|q_2^{\frac14}\Lambda).\label{qNYCShi}
		\end{aligned}
		\end{equation} 
		We have not found these relations in the literature but they follow from the results of \cite{NY09}.\footnote{These relations correspond to  $r=2$, $d=0,2$, $(c_1,[C])=1$ and $m=1$ in terms of \cite[(1.43)]{GNY06}. The proof is based on \cite[Thm. 2.11b)]{NY09}. Here it was assumed that $0<d<r$ but the proof works for $d=0,r$ as well, except the last argument based on the vanishing of $f_{2*}(\det \mathcal{S}^{\otimes d})$. Recall that here $f_2$ is the Grassmanian $Gr(n,r)$ bundle (due to \cite[Prop. 1.2]{NY09}) and  $\mathcal{S}$ is a universal rank $n$ bundle. Hence for $d=0,r$ the sheaf $f_{2*}(\det \mathcal{S}^{\otimes d})$ becomes $\mathcal{O}$ up to degree shift due to	Borel-Bott-Weil theorem for the Grassmannians.  We are grateful to H. Nakajima for the explanation on this point.}
		
		Returning to the case $q_1q_2=1$ we see that relations \eqref{qNYCS2} and \eqref{qNYCS3} coincide. We rewrite these relations at the level of CS-modified $c=-2$ tau functions (analogs of \eqref{qtaum2pm} with $\mathcal{Z}_1$ instead of $\mathcal{Z}$)  
		\begin{equation}
		\begin{aligned}
		2\tau_{1;0}(z)=\uptau^+(q^{1/2}z)\uptau^-(q^{-1/2}z)+\uptau^+(q^{-1/2}z)\uptau^-(q^{1/2}z)=\\
		=\uptau^+(q^{3/2}z)\uptau^-(q^{-3/2}z)+\uptau^+(q^{-3/2}z)\uptau^-(q^{3/2}z). \label{qNYDCS2diff}
		\end{aligned}
		\end{equation}
		Two relations in \eqref{qNYCShi} also coincide. In terms of tau functions we obtain
		\begin{equation}
		-2z^{1/4}\tau_{1;1}(z)=\uptau^+(qz)\uptau^-(q^{-1}z)-\uptau^+(q^{-1}z)\uptau^-(qz). \label{qNYDCS1diff}
		\end{equation}
		
		\begin{prop}
			Let $\uptau^{\pm}$ satisfy the second equality in \eqref{qNYDCS2diff}. Then  $\tau_{1;0}(z)$ and $\tau_{1;1}(z)$ defined by \eqref{qNYDCS2diff}, \eqref{qNYDCS1diff} correspondingly 
			satisfy Toda-like equation \eqref{eq:qTodaCS} for $m=1$,  $j=0$. \label{qdoubleconjCS}
		\end{prop}
		
		\begin{proof}
			Substituting \eqref{qNYDCS1diff}, \eqref{qNYDCS2diff} to \eqref{eq:qTodaCS} in this case,
			we obtain the equation
			\begin{equation}
			\begin{aligned}
			(\uptau^+(q^{3/2}z)\uptau^-(q^{1/2}z)+\uptau^+(q^{1/2}z)\uptau^-(q^{3/2}z)) (\uptau^+(q^{-1/2}z)\uptau^-(q^{-3/2}z)+\uptau^+(q^{-3/2}z)\uptau^-(q^{-1/2}z))=\\
			=(\uptau^+(q^{1/2}z)\uptau^-(q^{-1/2}z)+\uptau^+(q^{-1/2}z)\uptau^-(q^{1/2}z))(\uptau^+(q^{3/2}z)\uptau^-(q^{-3/2}z)+\uptau^+(q^{-3/2}z)\uptau^-(q^{3/2}z))-\\
			-(\uptau^+(q^{3/2}z)\uptau^-(q^{-1/2}z)-\uptau^+(q^{-1/2}z)\uptau^-(q^{3/2}z))(\uptau^+(q^{1/2}z)\uptau^-(q^{-3/2}z)-\uptau^+(q^{-3/2}z)\uptau^-(q^{1/2}z)).
			\end{aligned}
			\end{equation}
			To see that this identity holds, note that upon expanding the parentheses each summand
			is of the form $\uptau^{\eta_1}(q^{3/2}z)\uptau^{\eta_2}(q^{1/2}z)\uptau^{\eta_3}(q^{-1/2}z)\uptau^{\eta_4}(q^{-3/2}z)$, 
			where $\eta_{1,2,3,4}$ are signs and there are two signs "+" and two signs "-" in each summand. Therefore, we may label each summand by the
			positions of "+". In these notations, the previous relation is the identity
			\begin{equation}
		  (13)+(14)+(23)+(24)=(12)+(24)+(13)+(34)-(12)+(14)+(23)-(34). 
			\end{equation}
		\end{proof}
		Since we know that the functions $\uptau^\pm$ defined by CS-modified \eqref{qtaum2pm} and $\tau$ defined by CS-modified \eqref{Kiev} satisfy \eqref{qNYDCS2diff}, \eqref{qNYDCS1diff} then it follows from the above proposition that $\tau$ satisfy Toda-like equation \eqref{eq:qTodaCSsg} for $m=1$.\end{proof}
	
	\subsection{$q$-Painlev\'e $A_7^{(1)'}$ $c={-}2$ tau functions and $q$-Painlev\'e $A_3^{(1)}$ equation} 
	\label{ssec:clust2}
	Recall that $q$-Painlev\'e $A_3^{(1)}$ is a term for $q$-Painlev\'e VI equation. In this section, we study a surprising connection of $c=-2$ $A_7^{(1)'}$ tau functions (introduced above) to this equation.
	
	Let us rewrite equations \eqref{qNYD2diff}, \eqref{qNYD1diff} in the form
	\begin{align}
	\overline{\uptau_0^+}=\frac{\uptau_0^+\uptau_0^- -z^{1/4}\uptau_1^+\uptau_1^-}{\underline{\uptau_0^-}}, \label{cd0p}\\
	\overline{\uptau_0^-}=\frac{\uptau_0^+\uptau_0^- +z^{1/4}\uptau_1^+\uptau_1^-}{\underline{\uptau_0^+}}, \label{cd0m}
	\end{align}
	where we have introduced notations $\uptau_0^{\pm}=\uptau^{\pm}$, $\uptau_1^{\pm}(u,s|z)=s^{1/4}\uptau^{\pm}(uq,s|z)$.
	Also we have one more pair of equations, which are B\"acklund transformed equations \eqref{cd0p}, \eqref{cd0m}
	\begin{align}
	\overline{\uptau_1^+}=\frac{\uptau_1^+\uptau_1^- -z^{1/4}\uptau_0^+\uptau_0^-}{\underline{\uptau_1^-}} \label{cd1p},\\
	\overline{\uptau_1^-}=\frac{\uptau_1^+\uptau_1^- +z^{1/4}\uptau_0^+\uptau_0^-}{\underline{\uptau_1^+}}. \label{cd1m}
	\end{align}
	We have obtained a closed system of four $q$-difference equations of second order on the tuple $(\uptau_0^+,\uptau_0^-,\uptau_1^+,\uptau_1^-)$.
	This is the same as closed system of eight $q$-difference equations of first order on $(\uptau_0^+,\uptau_0^-,\uptau_1^+,\uptau_1^-,
	\underline{\uptau_0^+}, \underline{\uptau_0^-}, \underline{\uptau_1^+}, \underline{\uptau_1^-})$.
	Actually this system is a particular case of $q$-Painlev\'e VI equation. To show that, we will need
	bilinear (or tau) form of $q$-Painlev\'e VI equation.
	This form was basically introduced in \cite{TM06}, \cite{S98}. Since we are interested in the solutions of this equation, it is more convenient for us to follow the exposition of \cite{JNS17}. 
	
	The $q$-Painlev\'e VI equation in tau form is a system of eight first order $q$-difference equations 
	on the tuple $(\tau_1,\tau_2,\tau_3,\tau_4,\tau_5,\tau_6,\tau_7,\tau_8)$ 
	\cite[Eq.(3.16)-(3.23)]{JNS17}
	\begin{equation}
	\begin{aligned}
	&\tau_1\tau_2-t\tau_3\tau_4=(1-q^{-2\theta_t}t)\underline{\tau_5}\overline{\tau_6},\\
	&\tau_1\tau_2-q^{-2\theta_1}t\tau_3\tau_4=(1-q^{-2\theta_1}t)\tau_5\tau_6,\\
	&\tau_1\tau_2-\tau_3\tau_4=-q^{2\theta_t}(1-q^{-2\theta_1}t)\underline{\tau_7}\overline{\tau_8},\\
	&\tau_1\tau_2-q^{2\theta_t}\tau_3\tau_4=-q^{2\theta_t}(1-q^{-2\theta_t}t)\tau_7\tau_8,\\
	&\underline{\tau_5}\tau_6+tq^{-\theta_1-\theta_{\infty}+\theta_t-1/2}\underline{\tau_7}\tau_8=\underline{\tau_1}\tau_2,\\
	&\underline{\tau_5}\tau_6+tq^{-\theta_1+\theta_{\infty}+\theta_t-1/2}\underline{\tau_7}\tau_8=\tau_1\underline{\tau_2},\\
	&\underline{\tau_5}\tau_6+q^{\theta_0+2\theta_t}\underline{\tau_7}\tau_8=q^{\theta_t}\underline{\tau_3}\tau_4,\\
	&\underline{\tau_5}\tau_6+q^{-\theta_0+2\theta_t}\underline{\tau_7}\tau_8=q^{\theta_t}\tau_3\underline{\tau_4}, \label{qPVI}
	\end{aligned}
	\end{equation}
	where $t$ is an independent variable.  
	Let us assign to each tau function $\tau_i,\, i=1,\ldots 8$ a tuple ${\boldsymbol{\theta}}=(\theta_0,\theta_t,\theta_1,\theta_{\infty})$ as follows
\begin{center}
\begin{tabular}{|c|c|c|c|c|c|c|c|c|}
\hline Tau function & $\tau_1$ & $\tau_2$& $\tau_3$& $\tau_4$ &$\tau_5$&$\tau_6$& $\tau_7$ &$\tau_8$\\
\hline Tuple  & $\boldsymbol{\theta_{\infty}^{\uparrow}}$ & $\boldsymbol{\theta_{\infty}^{\downarrow}}$& $\boldsymbol{\theta_0^{\uparrow}}$& $\boldsymbol{\theta_0^{\downarrow}}$ &$\boldsymbol{\theta_1^{\downarrow}}$&$\boldsymbol{\theta_1^{\uparrow}}$& $\boldsymbol{\theta_t^{\downarrow}}$ &$\boldsymbol{\theta_t^{\uparrow}}$ \\
\hline
\end{tabular}\, , 
\end{center}
where $\boldsymbol{\theta_k^{\uparrow}}$ and $\boldsymbol{\theta_k^{\downarrow}}, \, k=0,t,1,\infty$ denotes that $(\boldsymbol{\theta_k^{\uparrow}})_l=\theta_l+\frac12\delta_{k,l}$ and $(\boldsymbol{\theta_k^{\downarrow}})_l=\theta_l-\frac12\delta_{k,l}$. For a few formulas below we will use notations like $\tau^{\uparrow}_{\infty}=\tau_1$, $(\boldsymbol{\theta^{\uparrow}_{\infty}})_l=\theta_{1,l}$ in parallel. 
	
	Let us now change the normalizations for convenience.  
	Denote the tau functions which satisfy the system \eqref{qPVI}  by $\tau_i^{JNS},\, i=1,\ldots 8$ and the independent  variable by $t^{JNS}$. Make the substitution  $t^{JNS}=q^{\theta_{t}+\theta_1} t$ and
	\begin{equation}
	\begin{aligned}
	\tau_i(t)&=\lmb(\boldsymbol{\theta_i})(q^{\theta_{i,t}+\theta_{i,1}}t)^{\theta_{i,0}^2+\theta_{i,t}^2}\prod_{\epsilon=\pm1}(q^{1+\epsilon(\theta_{i,1}-\theta_{i,t})}t;q;q)_{\infty} \tau_i^{JNS}(q^{\theta_{i,t}+\theta_{i,1}}t),&\, i&=1,\ldots4\\
	\tau_i(t)&=\lmb(\boldsymbol{\theta_i})(q^{\theta_{i,t}+\theta_{i,t}}q^{-\frac12}t)^{\theta_{i,0}^2+\theta_{i,t}^2}\prod_{\epsilon=\pm1}(q^{1+\epsilon(\theta_{i,1}-\theta_{i,t})}q^{-\frac12}t;q;q)_{\infty} \tau_i^{JNS}(q^{\theta_{i,t}+\theta_{i,1}}q^{-\frac12}t),&\, i&=5,\ldots 8 \label{substtau}
	\end{aligned}
	\end{equation}
	where the function $\lmb$ 
	should satisfy
	relations
	\begin{equation}
	    \begin{aligned}
	    \frac{\lmb(\boldsymbol{\theta^{\uparrow}_{0}})\lmb(\boldsymbol{\theta_{0}^{\downarrow}})}{\lmb(\boldsymbol{\theta^{\uparrow}_{\infty}})\lmb(\boldsymbol{\theta_{\infty}^{\downarrow}})}=q^{\frac12(\theta_t-\theta_1)},\quad 
	    \frac{\lmb(\boldsymbol{\theta_1^{\uparrow}})\lmb(\boldsymbol{\theta_1^{\downarrow}})}{\lmb(\boldsymbol{\theta_{\infty}^{\uparrow}})\lmb(\boldsymbol{\theta_{\infty}^{\downarrow}})}=1,\quad
	    \frac{\lmb(\boldsymbol{\theta_t^{\uparrow}})\lmb(\boldsymbol{\theta_t^{\downarrow}})}{\lmb(\boldsymbol{\theta_0^{\uparrow}})\lmb(\boldsymbol{\theta_0^{\downarrow}})}=1.
	    \end{aligned}
	\end{equation}
	The specific choice of $\lmb$ is not essential for the below considerations. We can take, for example,
	$\lmb(\boldsymbol{\theta})=q^{-\theta_1(\theta_0^2+\theta_t^2)-\theta_t(\theta_1^2+\theta_{\infty}^2)}$. Other possible choice, which is analytic in $q^{\theta_k},\,k=0,t,1,\infty$ is expressed in terms of the elliptic Gamma functions
	(see \eqref{ellGamma} and \eqref{eq:ellGamma_shift})
	\begin{equation}
	\lmb(\boldsymbol{\theta})^{-1}=\prod_{\epsilon=\pm1}\Gamma(q^{\frac14(\frac12+\theta_t+\epsilon(\theta_0+\frac14))};q^{\frac18},q^{\frac18})\Gamma(q^{\frac14(\frac12+\theta_1+\epsilon(\theta_{\infty}+\frac14))};q^{\frac18},q^{\frac18}). 
	\end{equation}
	Note that in the substitution \eqref{substtau} the argument in $\tau_i, i=1,\ldots 4,6,8$ is $q^{\theta_{t}+\theta_1} t$, 
	but in $\tau_i, i=5,7$ it differs by $q^{-1}$. 
	
	Then the system of equations \eqref{qPVI} transforms into the system
	\begin{equation}
	   \begin{aligned}
	   \tau^{\uparrow}_{\infty}\tau^{\downarrow}_{\infty}- q^{\theta_1}t^{1/2}\tau^{\uparrow}_0\tau^{\downarrow}_0=\overline{\tau_1^{\uparrow}}\tau_1^{\downarrow}\\
	   \tau^{\uparrow}_{\infty}\tau^{\downarrow}_{\infty}- q^{-\theta_1}t^{1/2}\tau_0^{\uparrow}\tau_0^{\downarrow}=\tau^{\uparrow}_1 \overline{\tau^{\downarrow}_1}\\
       \tau^{\uparrow}_0\tau^{\downarrow}_0-q^{\theta_t}t^{1/2}\tau^{\uparrow}_{\infty}\tau^{\downarrow}_{\infty}=\overline{\tau^{\uparrow}_{t}}\tau_t^{\downarrow}\\
       \tau^{\uparrow}_0\tau^{\downarrow}_0-q^{-\theta_t}t^{1/2}\tau^{\uparrow}_{\infty}\tau^{\downarrow}_{\infty}=\tau^{\uparrow}_t \overline{\tau^{\downarrow}_t}\\ 
       \tau^{\uparrow}_1\tau_1^{\downarrow}+ q^{-\theta_{\infty}-1/4}t^{1/2}\tau^{\uparrow}_t\tau_t^{\downarrow}=\underline{\tau_{\infty}^{\uparrow}}\tau_{\infty}^{\downarrow}\\
       \tau^{\uparrow}_1\tau_1^{\downarrow}+ q^{\theta_{\infty}-1/4}t^{1/2}\tau^{\uparrow}_t\tau_t^{\downarrow}=\tau_{\infty}^{\uparrow}\underline{\tau_{\infty}^{\downarrow}}\\
      \tau^{\uparrow}_t\tau_t^{\downarrow}+ q^{-\theta_0-1/4}t^{1/2} \tau^{\uparrow}_1\tau_1^{\downarrow}=\underline{\tau_0^{\uparrow}}\tau_0^{\downarrow}\\
      \tau^{\uparrow}_t\tau_t^{\downarrow}+ q^{\theta_0-1/4}t^{1/2} \tau^{\uparrow}_1\tau_1^{\downarrow}=\tau_0^{\uparrow}\underline{\tau_0^{\downarrow}}
       \label{qPVIm}
	   \end{aligned}
	\end{equation}
	We make this substitution mainly to remove non-monomial coefficients
	like $(1-q^{\ldots})$ from these equations and to make it more symmetric and natural. 
	We will use below only this normalization of the tau function. 
	\begin{prop}\label{VI2eq}
		Consider the tuple $(\tau_{\infty}^{\uparrow},\tau_{\infty}^{\downarrow},\tau_{0}^{\uparrow},\tau_0^{\downarrow},\tau_{1}^{\downarrow},\tau_1^{\uparrow},\tau_t^{\downarrow},\tau_t^{\uparrow})=(\uptau^+_0,\uptau^-_0,\uptau^+_1,\uptau^-_1,\underline{\uptau^+_0},\underline{\uptau^-_0},\underline{\uptau^+_1},\underline{\uptau^-_1)}$, where the functions $\uptau_0^{\pm}, \uptau_1^{\pm}$ satisfy \eqref{cd0p}, \eqref{cd0m}, \eqref{cd1p},\eqref{cd1m}. This tuple is a solution of 
		\eqref{qPVIm} in the case $q^{\theta_0}=q^{\theta_t}=q^{\theta_1}=q^{\theta_{\infty}}=i$
		under the identification $t^{1/2}=iz^{1/4}$. 
	\end{prop}
	\begin{proof}
		Let us substitute $t^{1/2}=iz^{1/4}$ to \eqref{qPVIm}.
		Note that this naturally leads to notation change $\overline{f(t)}=f(qt)=
		f(iqz^{1/2}):=f(-q^2z)=\overline{\overline{f(-z)}}$. Then we obtain  
		\begin{equation}
		\begin{aligned}
		&\tau^{\uparrow}_{\infty}\tau^{\downarrow}_{\infty}+z^{1/4}\tau^{\uparrow}_0\tau^{\downarrow}_0=\tau_1^{\downarrow}\overline{\overline{\tau_1^{\uparrow}}},\qquad
		&\tau^{\uparrow}_{\infty}\tau^{\downarrow}_{\infty}-z^{1/4}\tau^{\uparrow}_0\tau^{\downarrow}_0=\overline{\overline{\tau_1^{\downarrow}}}\tau_1^{\uparrow},\\
		&\tau^{\uparrow}_{0}\tau^{\downarrow}_{0}+z^{1/4}\tau^{\uparrow}_{\infty}\tau^{\downarrow}_{\infty}=\tau_t^{\downarrow}\overline{\overline{\tau_t^{\uparrow}}},\qquad
		&\tau^{\uparrow}_{0}\tau^{\downarrow}_{0}-z^{1/4}\tau^{\uparrow}_{\infty}\tau^{\downarrow}_{\infty}=\overline{\overline{\tau_t^{\downarrow}}}\tau_t^{\uparrow},\\
		&\overline{\tau_1^{\downarrow}}\overline{\tau_1^{\uparrow}}+z^{1/4}\overline{\tau_t^{\downarrow}}\overline{\tau_t^{\uparrow}}=\underline{\tau_{\infty}^{\uparrow}}\overline{\tau_{\infty}^{\downarrow}},\qquad
		&\overline{\tau_1^{\downarrow}}\overline{\tau_1^{\uparrow}}-z^{1/4}\overline{\tau_t^{\downarrow}}\overline{\tau_t^{\uparrow}}=\overline{\tau_{\infty}^{\uparrow}}\underline{\tau_{\infty}^{\downarrow}},\\
		&\overline{\tau_t^{\downarrow}}\overline{\tau_t^{\uparrow}}+z^{1/4}\overline{\tau_1^{\downarrow}}\overline{\tau_1^{\uparrow}}=\underline{\tau_{0}^{\uparrow}}\overline{\tau_{0}^{\downarrow}},\qquad
		&\overline{\tau_t^{\downarrow}}\overline{\tau_t^{\uparrow}}-z^{1/4}\overline{\tau_1^{\downarrow}}\overline{\tau_1^{\uparrow}}=\overline{\tau_{0}^{\uparrow}}\underline{\tau_{0}^{\downarrow}},
		\end{aligned} 
		\end{equation}
		where we have also transformed $z\mapsto qz$ in the last four equations.
		We can take the ansatz $(\tau_1^{\downarrow},\tau_1^{\uparrow},\tau_t^{\downarrow},\tau_t^{\uparrow})=(\underline{\tau_{\infty}^{\uparrow}},\underline{\tau_{\infty}^{\downarrow}},\underline{\tau_0^{\uparrow}},\underline{\tau_0^{\downarrow}})$
		under which the first four equations and the last became equivalent.
		But these are just equations \eqref{cd0p}, \eqref{cd0m}, \eqref{cd1p}, \eqref{cd1m} on
		$(\tau_{\infty}^{\uparrow},\tau_{\infty}^{\downarrow},\tau_0^{\uparrow},\tau_0^{\uparrow})=(\uptau^+_0,\uptau^-_0,\uptau^+_1,\uptau^-_1)$.
		Therefore, we obtain a solution of the $q$-Painlev\'e VI \eqref{qPVIm}.
	\end{proof}
	\begin{Remark}Note that the dynamics of $c=-2$ tau functions considered as $q$-Painlev\'e VI tau functions is a "square root" of the standard $q$-Painlev\'e VI dynamics. This is manifested in the relation $t=-z^{1/2}$. This "square root" belongs to the full symmetry group of the $q$-Painlev\'e VI equation but does not belong to the normal subgroup of translations. But for special values of parameters $q^\theta_i$ this "square root" can be viewed as a dynamics.
	
	The relation between $q$-Painlev\'e $A_7^{(1)'}$ equation and $q$-Painlev\'e $A_3^{(1)}$ equation is similar to the folding transformation \cite{TOS05}. We hope to discuss this elsewhere.
	\end{Remark}
	It was proposed in \cite{JNS17} that solutions of $q$-Painlev\'e VI can be written in terms of 
	a single tau function $\tau(\boldsymbol{\theta};\sg,s|z)$. This tau function is given by \eqref{Kiev}
	where partition function is 5d Nekrasov function with $4$ matter fields and with condition $q_1q_2=1$ 
	\begin{equation}
	\tau(\boldsymbol{\theta};\sg,s|t)
	=\lmb(\boldsymbol{\theta})\sum_{n\in\mathbb{Z}}s^n (q^{\theta_t+\theta_1}t)^{(\sg+n)^2}\times C(\boldsymbol{\theta}|\sg+n)
	\times\mathcal{Z}_{tv}(\boldsymbol{\theta},\sg+n|t).\label{qPVItau}
	\end{equation}
	The first multiplier here is $\mathcal{Z}_{cl}$,
	which differs from the standard $t^{(\sg+n)^2-\theta_0^2-\theta_t^2}$
	by the factor $t^{\theta_0^2+\theta_t^2}$, which agrees with the 
	substitution \eqref{substtau}.
	The second multiplier is
	$\mathcal{Z}_{1-loop}$ given by
	\begin{equation}
	\begin{aligned}
	&C(\boldsymbol{\theta}|\sg)=
	\frac{\prod\limits_{\epsilon,\epsilon'=\pm1}
		(q^{1+\epsilon\theta_{\infty}-\theta_1+\epsilon'\sg};q,q)_{\infty}(q^{1+\epsilon\theta_{0}-\theta_t+\epsilon'\sg};q,q)_{\infty}}
	{(q^{1-2\sg};q,q)_{\infty}(q^{1+2\sg};q,q)_{\infty}}.
	\end{aligned}
	\end{equation}
	The third multiplier $\mathcal{Z}_{tv}$ is given by the product of two $q$-Pochhammer symbols and the instanton part $\mathcal{Z}_{inst}$, given by the standard expression
	\begin{equation}
	\begin{aligned}
	&\mathcal{Z}_{tv}(\boldsymbol{\theta},\sg|t)=\prod_{\epsilon=\pm1}(q^{1+\epsilon(\theta_1-\theta_t)}t;q;q)_{\infty} \mathcal{Z}_{inst}(\boldsymbol{\theta},\sg|q^{\theta_t+\theta_1}t),\\
    &\mathcal{Z}_{inst}(\boldsymbol{\theta},\sg|t)=\sum_{\lmb^{(1)},\lmb^{(2)}}t^{|\lmb^{(1)}|+|\lmb^{(2)}|}
	\prod_{\epsilon,\epsilon'=\pm1}\frac{\mathsf{N}_{\emptyset,\lmb^{(\epsilon')}}(q^{\epsilon\theta_{\infty}-\theta_1-\epsilon'\sg};q^{-1},q)
		\mathsf{N}_{\lmb^{(\epsilon')},\emptyset}(q^{\epsilon'\sg-\theta_t-\epsilon\theta_0};q^{-1},q)}
	{\mathsf{N}_{\lmb^{(\epsilon)},\lmb^{(\epsilon')}}(q^{(\epsilon-\epsilon')\sg};q^{-1},q)}.\\
	\end{aligned}
	\end{equation}
	Note that the normalization of the 5d instanton
	partition function given by $\mathcal{Z}_{tv}$ is the same as the normalization arising from the topological vertex approach \cite[(3.47)]{MPTY14}. This normalization is consistent with the
	substitution \eqref{substtau}.
	There is a conjecture that $\mathcal{Z}_{tv}$ has
	$SO(8)$ symmetry. Below we will use the symmetries of the form
	\begin{align}
	&\mathcal{Z}_{tv}(\theta_0,\theta_t,\theta_1,\theta_{\infty},\sg|t)=\mathcal{Z}_{tv}(\theta_t,\theta_0,\theta_{\infty},\theta_{1},\sg|t),  \label{pchs} \\
	&\mathcal{Z}_{tv}(\theta_0,\theta_t,\theta_{1},\theta_{\infty},\sg|t)=\mathcal{Z}_{tv}(\theta_{\infty},\theta_1,\theta_{t},\theta_{0},\sg|t).
	\label{dpairtr}
	\end{align}

	\begin{conj}\cite{JNS17}
		Eight tau functions
		\begin{equation}
		\begin{aligned}
		&\tau^{\uparrow}_{\infty}=\tau(\boldsymbol{\theta^{\uparrow}_{\infty}};\sg,s|t),&\qquad &\tau^{\downarrow}_{\infty}=\tau(\boldsymbol{\theta^{\downarrow}_{\infty}};\sg,s|t),\\
		&\tau^{\uparrow}_{0}=\tau(\boldsymbol{\theta^{\uparrow}_{0}};\sg+1/2,s|t),&\qquad &\tau^{\downarrow}_{0}=\tau(\boldsymbol{\theta^{\downarrow}_{0}};\sg-1/2,s|t),\\
		&\tau^{\downarrow}_{1}=\tau(\boldsymbol{\theta^{\downarrow}_{1}};\sg,s|q^{-\frac12}t),&\qquad &\tau^{\uparrow}_{1}=\tau(\boldsymbol{\theta^{\uparrow}_{1}};\sg,s|q^{-\frac12}t),\\
		&\tau^{\downarrow}_{t}=\tau(\boldsymbol{\theta^{\downarrow}_{t}};\sg+1/2,s|q^{-\frac12}t),&\qquad &\tau^{\uparrow}_{t}=\tau(\boldsymbol{\theta^{\uparrow}_{t}};\sg-1/2,s|q^{-\frac12}t)
		\end{aligned}
		\end{equation}
		satisfy $q$-Painlev\'e VI equation in the tau form \eqref{qPVIm}.\label{JNSconj}
	\end{conj}
	
	The $q$-Painlev\'e VI equation is the difference equation of second order, so it is natural to expect that up to $q$-periodicity general  solution of the $q$-Painlev\'e VI equation belong to the two-parameter family given by the Conjecture \ref{JNSconj}.
	Hence for the case $q^{\theta_0}=q^{\theta_t}=q^{\theta_1}=q^{\theta_{\infty}}=i$ due to the Proposition \ref{VI2eq} we expect that
	there is a connection between the $c=-2$ tau functions given by \eqref{qtaum2pm} and the $c=1$ tau functions given by the \eqref{qPVItau}. 
	However this connection may be not just an equality of these tau functions with some parameters
	$s, \sg$ and $\td{s}, \td{\sg}$. It is because equations \eqref{qPVIm} have symmetry
	$\tau^{\epsilon}_k\mapsto f(u)\tau_k^{\epsilon},\, k=0,t,1,\infty,\,\epsilon=\uparrow,\downarrow$, where the function $f(u)$ is $q$-periodic in $u$ and the symmetry $\tau^{\uparrow}_k\mapsto h_k(u)\tau^{\uparrow}_{k},\tau^{\downarrow}_{k}\mapsto h_k^{-1}(u)\tau^{\downarrow}_{k},\, k=0,t,1,\infty$, where functions $h_k(u)$ are
	also $q$-periodic in $u$.
	
	Moving to the multiplicative notations $v_i= q^{\theta_i}, u=q^{2\sg}$ we have
	\begin{conj}\label{dynequal}
		There exist such functions $\td{u}=\td{u}(u)$, $\td{s}=\td{s}(s,\sg)$
		and $q$-periodic in $z$ functions $f(u;q)$, $h_k(u;q)$, $k=0,t,1,\infty$ such that  
		\begin{align}
		\tau(i,i,i,iq^{\pm1/2};\td{u},\td{s}|-z^{1/2})&=f(u;q)h_{\infty}^{\pm1}(u;q)\uptau_0^{\pm}(u,s|z),\\
		\tau(iq^{\pm1/2},i,i,i;\td{u}q^{\pm1},\td{s}|-z^{1/2})&=f(u;q)h_0^{\pm1}(u;q)\uptau_1^{\pm}(u,s|z),\\
		\tau(i,i,iq^{\mp1/2},i;\td{u},\td{s}|-q^{-1/2}z^{1/2})&=f(u;q)h_1^{\mp1}(u;q)\uptau_0^{\pm}(u,s|q^{-1}z),\\
		\tau(i,iq^{\mp1/2},i,i;\td{u}q^{\pm1},\td{s}|-q^{-1/2}z^{1/2})&=f(u;q)h_t^{\mp1}(u;q)\uptau_1^{\pm}(u,s|q^{-1}z).
		\end{align}
	\end{conj}
	Below we make this conjecture more precise, namely, give formulas for
	$\td{u}$, $\td{s}$ and functions $f(u;q)$, $h_k(u;q)$, $k=0,t,1,\infty$.
	
	Comparing the powers of $z$ in 
	any of the $4$ equalities of the conjecture we obtain from
	\eqref{qtaum2pm} and
	\eqref{qPVItau} that it is necessary to take $\td{\sg}=\sg$.
	
	The terms in the sums in l.h.s. and r.h.s. with $\mathcal{Z}_{cl}=z^{\frac12(\sg+n)^2}$
	are linearly independent for different $n\in\mathbb{Z}$.
	Therefore, at the next step we compare the coefficients in front of powers $z^{\frac12(\sg+n)^2}$ in the all $4$ equalities of the conjecture.
	
	When we compare it up to the $z$-constant functions we obtain that due to \eqref{dpairtr} it is necessary to take
	\begin{equation}
	\begin{aligned}
	&\mathcal{Z}_{tv}(i,i,i,iq^{\pm1/2},u;q^{-1},q|-z^{1/2})=\mathcal{Z}_{inst}(u;q^{-1},q^2|z),\\
	&\mathcal{Z}_{tv}(i,i,iq^{\pm1/2},i,u;q^{-1},q|-q^{-1/2}z^{1/2})=\mathcal{Z}_{inst}(u;q^{-1},q^2|q^{-1}z),
	\end{aligned}
	\end{equation}
	where we start to write the dependence on $q_1,q_2$ explicitly. However, shifting $z\mapsto q z$ in the second relation we see that it is equivalent to the first one due to symmetry \eqref{pchs}. We have checked by the computer calculation up to order $z^{5}$ analytically that the first relation is satisfied. I.e. we have
	\begin{conj}\label{conj:prdx}
		There is a relation between pure 5d Nekrasov instanton partition function
		with $\epsilon_2=-2\epsilon_1$ and 5d Nekrasov instanton partition function with $\epsilon_2=-\epsilon_1$ and special values of $v_i$, namely
		\begin{equation}
		(-qz^{1/2};q,q)^2_{\infty}\mathcal{Z}_{inst}(i,i,i,iq^{\pm1/2},u;q^{-1},q|z^{1/2})=\mathcal{Z}_{inst}(u;q^{-1},q^2|z).\label{prdx}    
		\end{equation}
	\end{conj}
	
	Finally we compare the $z$-independent coefficients in front of power $z^{\frac12(\sg+n)^2}$.
	To process this step we use
	\begin{equation}
	\begin{aligned}
	&C(iq^{\pm1/2},i,i,i|\sg)=C(i,i,i,iq^{\pm1/2}|\sg)=
	\frac{1}{\theta(\pm q^{1/2+\sg};q)\prod_{\epsilon'=\pm1}(q^{2+2\epsilon'\sg};q,q^2)_{\infty}},\\
	&C(i,iq^{\pm1/2},i,i|\sg)=C(i,i,iq^{\pm1/2},i|\sg)=
	\prod_{\epsilon'=\pm1}\frac{1}{(q^{3/2\pm1/2+2\epsilon'\sg};q,q^2)_{\infty}}.
	\end{aligned}
	\end{equation}
	Using these formulas we obtain that it is necessary to take
	$\td{s}=s^{1/2}$, $f(u;q)=1$ and
	\begin{equation}
	\begin{aligned}
	&h_{\infty}(u;q)=\mu_{\infty}\theta(-u^{1/2}q^{1/2};q),&\quad &h_1(u;q)=\mu_1,\\
	&h_0(u;q)=\mu_{0}s^{-1/4}\theta(-u^{1/2}q;q),&\quad &h_t(u;q)=\mu_t s^{1/4},
	\end{aligned}
	\end{equation}
	where $\mu_{k}=\frac{\lmb(\boldsymbol{\theta_k^{+}})}{\lmb(\boldsymbol{\theta})}|\{v_l=i\},\,k,l=0,t,1,\infty$.
	Note that $h_1(u;q)$ is periodic in $\sg$ with period $1$.
	
	We have checked that Conjecture \ref{dynequal} follows from Conjecture \ref{conj:prdx}
	under the formulas for $\td{u}$, $\td{s}$ and $f(u;q)$, $h_k(u;q)$, $k=0,t,1,\infty$ given above.
	
	\begin{Remark}
		Note that $q$-Painlev\'e VI equation has the cluster nature (see \cite{BGM17}). This means that the dynamics described by this equation is given by a composition of mutations and permutations of vertices for the first quiver at the Fig. \ref{quiver}. As we have seen above
		$c=-2$ $q$-Painlev\'e $\mathrm{III}'_3$ tau  functions are a special case of $c=1$ $q$-Painlev\'e VI tau functions, 
		so they also admit cluster dynamics with the same quiver. To restore positivity the signs "-" in \eqref{cd0p}, \eqref{cd0m}, \eqref{cd1p}, \eqref{cd1m} are hidden into the cluster coefficcients.
		Note that $q$-Painlev\'e VI quiver and $q$-Painlev\'e~$\mathrm{III}'_3$ quiver (the second quiver at the Fig. \ref{quiver}) have much in common from the symmetry point of view.
		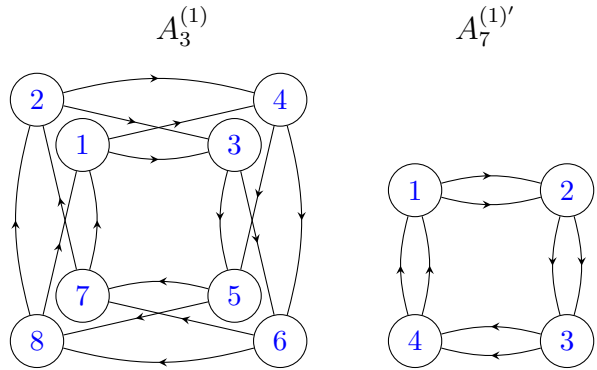
\begin{figure}[ht]
			\begin{center}
				\begin{tabular}{c c}
					$A_3^{(1)}$ & $A_7^{(1)'}$ \vspace{0.3cm}\\
					\begin{tikzpicture}[scale=2, font = \small]
					\node[shape=circle,draw=black] (A) at (0,1) {\color{blue}1};
					\node[shape=circle,draw=black] (A1) at(-0.3,1.3) {\color{blue}2};
					\node[shape=circle,draw=black] (B) at (1,1) {\color{blue}3};
					\node[shape=circle,draw=black] (B1) at(1.3,1.3) {\color{blue}4};
					\node[shape=circle,draw=black] (C) at (1,0) {\color{blue}5};
					\node[shape=circle,draw=black] (C1) at(1.3,-0.3) {\color{blue}6};
					\node[shape=circle,draw=black] (D) at (0,0) {\color{blue}7};
					\node[shape=circle,draw=black] (D1) at(-0.3,-0.3) {\color{blue}8};
					\path
					\foreach \from/\to in {A/B1,A1/B,B/C1,B1/C,C/D1,C1/D,D/A1,D1/A}
					{
						(\from) edge[mid arrow] (\to)					
					}
					\foreach \from/\to in {A/B,B/C,C/D,D/A}
					{
						(\from) edge[mid arrow,bend right=15] (\to)					
					}
					\foreach \from/\to in {A1/B1,B1/C1,C1/D1,D1/A1}
					{
						(\from) edge[mid arrow,bend left=15] (\to)					
					}
					;
					\end{tikzpicture}
					\hspace{0.5cm}
					&
					\begin{tikzpicture}[scale=2, font = \small]
					\node[shape=circle,draw=black] (A) at (0,1) {\color{blue}1};
					\node[shape=circle,draw=black] (B) at (1,1) {\color{blue}2};
					\node[shape=circle,draw=black] (C) at (1,0) {\color{blue}3};
					\node[shape=circle,draw=black] (D) at (0,0) {\color{blue}4};
					\path \foreach \from/\to in {A/B,B/C,C/D,D/A}
					{
						(\from) edge[mid arrow,bend right=15] (\to)
						(\from) edge[mid arrow,bend left=15] (\to)
					};			
					\end{tikzpicture}	
				\end{tabular}
			\end{center}
			\caption{Quiver of $q$-Painlev\'e VI and $q$-Painlev\'e $\mathrm{III}'_3$}
			\label{quiver}
		\end{figure}
	\end{Remark}
	
	\begin{Remark}
		Note that the formula \eqref{prdx} reflects the property that the coefficients with half-integer powers of $z$
		in the l.h.s. vanishes. We have checked by the computer calculations up to $z^{7/2}$ analytically that this property is satisfied in a more general situation, namely ($\alpha\in\mathbb{C}$)
		\begin{equation}
		(-qz;q,q)^2_{\infty}\mathcal{Z}_{inst}(i,i,i,i\alpha,u;q^{-1},q|z)=f(z^2).  
		\end{equation} 
	\end{Remark}
	
	\begin{Remark}
		We are not aware of any continuous analog of the relation \eqref{prdx}. In particular,
		in the $R\rightarrow0$ limit for the pure Nekrasov function we rescale $z\mapsto R^4z$ but in the case of Nekrasov function with $4$ matter fields $z$ is not rescaled.
		Also, $v_j=q^{\theta_j}$ in this limit should go to $1$ without rescaling $\theta_j$,
		but here all $v_j$ go to $i$.
	\end{Remark}
	
	\subsection{Connection with ABJ theory}\label{ssec:ABJM}
	
	Bonelli, Grassi, Tanzini in the paper \cite{BGT17} have proposed a version of the expression \eqref{Kiev} for the tau function of $q$-Painlev\'e $A_7^{(1)'}$ equation. We denote this tau function by $\tau_{\mathrm{BGT}}$. Contrary to the case studied in the present paper the formulas in \cite{BGT17} work only in the case $|q|=1$, therefore the function $\mathcal{Z}$ should be redefined by adding certain (non-perturbative) corrections, this is~a choice of another function $C$ in terms of Remark \ref{CRemark}. The function $\tau_{\mathrm{BGT}}$ depends only on one parameter, in terms of the formula \eqref{Kiev} this corresponds to the case $s=1$.
	
	By the topological string/spectral theory duality conjecture \cite{GHM14} the function $\tau_{\mathrm{BGT}}$ essentially equals to a spectral determinant of the operator 
	\begin{equation}
	\rho=(e^{\hat{p}}+e^{-\hat{p}}+e^{\hat{x}}+me^{-\hat{x}})^{-1}. 
	\end{equation}
	Here the operators $\hat{x},\hat{p}$ satisfy the commutation relation $[\hat{x},\hat{p}]=i\hbar$. Therefore, the operator $\rho$ is the inverse of the Hamiltonian of the affine relativistic Toda chain on two sites. 
	The relation between the parameters of the Hamiltonian and parameters of $\tau_{\mathrm{BGT}}$ are given by
	\begin{equation}
	\hbar = \frac{4\pi^2 i}{\log q},\quad m=\exp\left(\frac{-\hbar \log z }{2\pi}\right).
	\end{equation}
	
	Denote by $\Xi(\kappa,z)=\det(1+\kappa\rho)$ a spectral (Fredholm) determinant of the operator $\rho$. In terms of $\tau_{\mathrm{BGT}}$, the parameter $\kappa$ is expressed through $z,u,q$ by a quantum mirror map, see \cite{BGT17} for details. The topological string/spectral theory duality conjecture in this case means that
	\begin{equation}
	\tau_{\mathrm{BGT}}(u|z)=Z_{\rm CS} (z)\Xi(\kappa,z).
	\end{equation}
	The auxiliary function $Z_{\rm CS} $ is given in \cite{BGT17} by an explicit expression and satisfies
	\begin{equation}\label{eq:Z_CS:eq}
	\overline{Z_{\rm CS} ( z)}\underline{ Z_{\rm CS} (z)}=(z^{1/4}+z^{-1/4})Z^2_{\rm CS} (z).
	\end{equation}
	The function $Z_{\rm CS} $ is an analog of the algebraic solution, in the special case $\kappa=0$ we have $\tau_{\mathrm{BGT}}(z)=Z_{\rm CS} $. Note that the difference equation on our algebraic solution \eqref{eq:q-algebraic} has the form similar to \eqref{eq:Z_CS:eq}, but with $(1\mp z^{1/2})$ instead of $(z^{1/4}+z^{-1/4})$ in the r.h.s. 
	This is just a difference in normalization, it will not be important for our discussion in this section. 
	
	In the special case $z=q^M$, $M\in \mathbb{Z}$ the spectral determinant of the operator $\rho$ simplifies and equals to the grand canonical partition function of the ABJ theory. The parameter $M$ coincides with the difference of the ranks of two simple factors in the gauge group $U(N)\times U(N+M)$. In this case an interesting feature, the so-called Wronskian-like relations, were observed in \cite{GHM14'}. The function $\Xi(\kappa,z)$ can be factorised according to the parity
	of the eigenvalues of $\rho$, namely
	\begin{equation}
	\Xi(\kappa,z)=\Xi^+(\kappa,z)\Xi^-(\kappa,z).
	\end{equation}
	It was conjectured in \cite{GHM14'} that functions $\Xi^+, \Xi^-$ satisfy additional relations, which in our notations have the form
	\begin{equation}\label{eq:Xi+-:rel}
	\begin{aligned}
	iz^{1/4}\overline{\Xi_1^+} \underline{\Xi_1^-}-
	\overline{\Xi^+} \underline{\Xi^-}=(iz^{1/4}-1)\Xi^+\Xi^-,\\ 
	iz^{1/4}\underline{\Xi_1^+} \overline{\Xi_1^-}+
	\underline{\Xi^+} \overline{\Xi^-}=(iz^{1/4}+1)\Xi^+\Xi^-.
	\end{aligned}
	\end{equation}
	As before, $\Xi_1$ stands for B\"acklund transformation of  $\Xi$, in terms of the parameter $\kappa$ it is given by the	map $\kappa \rightarrow -\kappa$. This conjecture in \cite{GHM14'} was based on numerical checks.
	
	We conjecture that there exists a relation between  $\Xi^+, \Xi^-$ and $\uptau_{\mathrm{BGT}}^+, \uptau_{\mathrm{BGT}}^-$ introduced by an analog of the formula \eqref{qtaum2pm}. This conjecture could be viewed as a refinement of the topological string/spectral theory duality to the case of refined strings with parameters $t=q^2$. This relation is one of the main motivations of our paper.
	
	This conjecture is supported by a fact that multiplying equations \eqref{eq:Xi+-:rel} by appropriate auxiliary functions $Z_{\rm CS} ^+, Z_{\rm CS} ^-$, we obtain equations \eqref{qNYD2diff}, \eqref{qNYD1diff}. 
	Take auxiliary functions $Z_{\rm CS} ^{\pm}$ satisfying
	\begin{equation}
	\begin{aligned}
	\overline{Z_{\rm CS} ^+}\underline{Z_{\rm CS} ^-}=(1+iz^{1/4})Z_{\rm CS} ^+Z_{\rm CS} ^-,\\
	\underline{Z_{\rm CS} ^+}\overline{Z_{\rm CS} ^-}=(1-iz^{1/4})Z_{\rm CS} ^+Z_{\rm CS} ^-, 
	\end{aligned}
	\end{equation}
	and $Z_{\rm CS} ^+ Z_{\rm CS} ^-=Z_{\rm CS}$.
	Define $\uptau^{\pm}$ by $\Xi^{\pm}=Z_{\rm CS} ^{\pm}\uptau^{\pm}$,  then the functions $\uptau^{\pm}$ satisfy $\tau=\uptau^+\uptau^-$ and
	\begin{equation}
	\begin{aligned}
	iz^{1/4}\overline{\uptau_1^+}\underline{\uptau_1^-}-
	\overline{\uptau^+}\underline{\uptau^-}&=-(1+z^{1/2})\uptau^+\uptau^-,&\quad iz^{1/4}\overline{\uptau^+}\underline{\uptau^-}-
	\overline{\uptau_1^+}\underline{\uptau_1^-}&=-(1+z^{1/2})\uptau_1^+\uptau_1^-,\\
	iz^{1/4}\underline{\uptau_1^+}\overline{\uptau_1^-}+
	\underline{\uptau^+}\overline{\uptau^-}&=(1+z^{1/2})\uptau^+\uptau^-,&\quad 
	iz^{1/4}\underline{\uptau^+}\overline{\uptau^-}+
	\underline{\uptau_1^+}\overline{\uptau_1^-}&=(1+z^{1/2})\uptau_1^+\uptau_1^-,
	\end{aligned}
	\end{equation}
	where we also wrote B\"acklund transformed pair of equations.
	It is easy to see that these equations are equivalent to the equations \eqref{cd0p},\eqref{cd0m},\eqref{cd1p},\eqref{cd1m} (up to the change $z^{1/4} \mapsto -i z^{1/4}$ of the root branch). Note that the product $Z_{\rm CS}=Z_{\rm CS}^+Z_{\rm CS}^-$ satisfies a difference equation of the form \eqref{eq:Z_CS:eq} with the factor $(1+z^{1/2})$ since we work here in our normalization. The functions $Z_{\rm CS}^+, Z_{\rm CS}^-$ could be viewed as algebraic $c=-2$ tau functions constructed at the end of the Section \ref{ssec:qtaum2}.

	\section{Discussion}
	\label{sec:discuss}
	\paragraph{Painlev\'e VI.} In this paper we restrict ourselves to the parameterless Painlev\'e III and $q$-Painlev\'e III equations. It is natural to ask for generalizations to other Painlev\'e and $q$-Painlev\'e equations. 
	It looks like these generalizations should exist. 
	
	Consider, for example, Painlev\'e VI equation. One can define long $c=-2$ tau functions by the formulas similar to \eqref{taum201}, namely
	\begin{equation}\label{eq:uptauP6}
	\uptau(\boldsymbol{\theta};\sg,s|z)=\sum_{n\in\mathbb{Z}} s^n\mathcal{Z}_{c=-2}(\boldsymbol{\theta},\sg+2n|z).
	\end{equation}
	The simplest of the Nakajima-Yoshioka blowup relations in this case leads to the algebraic equation
	\begin{equation}
	\tau(\boldsymbol{\theta};\sg,s|z)=\uptau(\boldsymbol{\theta}+\frac12e_{23};\sg,s|z)\uptau(\boldsymbol{\theta}-\frac12e_{23};\sg,s|z)+\uptau(\boldsymbol{\theta}+\frac12e_{23};\sg+1,s|z)\uptau(\boldsymbol{\theta}-\frac12e_{23};\sg-1,s|z),\label{quadr3VI}
	\end{equation}
	where $e_{23}=(0,1,1,0)$ and
	$\tau$ is the Painlev\'e VI $c=1$ tau function. We obtained last relation just similarly to the case of Painlev\'e III($D_8^{(1)}$) discussed in the paper.

	
	There also exist analogous differential relations on these tau functions $\uptau$. One can deduce from them Toda-like equation on Painlev\'e VI $c=1$ tau function similarly to the Proposition~\ref{doubleprop}.
	
	\paragraph{Determinants and pfaffians.} Another approach to the Isomonodromy/CFT correspondence was proposed in \cite{MM17}. It was argued in loc. cit. that for resonant values of $\boldsymbol{\theta}$ and $\sigma$ the sum in the formula \eqref{Kiev}  becomes finite. In this case the answer is given by the Hankel determinant consisting of solutions of hypergeometric equations. The argument is based on insertion of screening operators and goes as in $\beta=2$ matrix models.
	
	For $c=-2$ case the insertion of screening operators leads to matrix models with $\beta=1$ or $\beta=4$, depending on the choice of the screening. To be more precise, the long tau function \eqref{eq:uptauP6} in the resonant case corresponds to $\beta=4$ and short tau function corresponds to $\beta=1$. In each case the tau function in the resonant case equals to a pfaffian.  We plan to discuss this elsewhere.
	
	\paragraph{Higher rank generalization.} Isomonodromy/CFT correspondence exists for any rank, as well as Nakajima-Yoshioka blowup relations. So it is natural to ask for such generalization of the results discussed in this paper. 
	
	For example, one can ask for the proof of the conjecture \cite{BGM18} of the solution of deautonomized Toda flows in terms of pure 5d $SU(N),\, N>2$ Nekrasov partition functions with Chern-Simons term.  Many blowup relations for this case were written in \cite{GNY06}, \cite{NY09}, but it seems that the conjecture of \cite{BGM18} can not be derived from them in a simple way. Note that this conjecture is proved for the analog of the algebraic solution in \cite{BG19}.
	
	\paragraph{Quantization.} It was conjectured in \cite{BGM17} that (Fourier) series of Nekrasov partition functions for generic $\epsilon_1,\epsilon_2$ satisfy quantum Painlev\'e equation in tau form. This conjecture is based on bilinear relations \eqref{eq:Z2=ZZ} for generic $\epsilon_1,\epsilon_2$. 
	
	A related problem is to deduce this to $c=-2$ tau functions introduced in this paper? Another related question, is it possible to deduce the relations \eqref{eq:Z2=ZZ} from usual Nakajima-Yoshioka blowup relations, in this paper we did this only for $\epsilon_1+\epsilon_2=0$. A third question is to rewrite Nakajima-Yoshioka blowup relations as equations on quantum tau functions.
	
	\paragraph{Riemann-Hilbert problem.} As was mentioned in the Introduction, arguing similarly to \cite{ILT14} one can construct matrix $Y(y)$ with prescribed monodromies which consists of sums of $c=-2$ Virasoro conformal blocks. These conformal blocks should include two degenerate fields at the points $y_0,y$ and $n$ primary fields in points $a_1,\ldots, a_n$. In order to formulate  Riemann-Hilbert problem completely it is necessary to specify behaviour of $Y(y)$ near $y_0$ and $a_i$. This seems to be an interesting open problem.
	

	\appendix
	
	\section{Some special functions}
	
	\label{sec:q}
	
	\paragraph{$q$-Pochhammer symbols and related $q$-special functions.}
	Here we collect some facts about $q$-special functions used in the paper. For the references
	about $q$-Pochhammer symbols and related functions, see \cite[Sec. 10]{AAR99}.
	
	Multiple infinite $q$-Pochhammer symbol is defined by
	\begin{equation}
	(z;q_1,\ldots q_N)_{\infty}=\prod_{i_1,\ldots i_N=0}^{\infty}\left(1-z\prod_{k=1}^Nq_k^{i_k}\right).
	\label{Pochhammer_def}
	\end{equation}
	This function is symmetric with respect to $q_k$.
	The product is well defined for arbitrary $z$ when $|q_k|<1$ for all $k$. 
	From the definition we obtain shift relations for the $q$-Pochhammer symbol
	\begin{equation}
	(z;q_1,\ldots q_N)_{\infty}/(zq_1;q_1,\ldots q_N)_{\infty}=(z;q_2,\ldots q_N)_{\infty}, \quad (z;q)_{\infty}/(zq;q)_{\infty}=1-z \label{qshift} 
	\end{equation}
	and also $n$-period relation 
	\begin{equation}
	\prod_{i=0}^{n-1}(zq_1^i;q_1^n,q_2,\ldots q_N)_{\infty}=(z;q_1,q_2,\ldots q_N)_{\infty}.\label{permult} 
	\end{equation}
	We also use the formula 
	\begin{equation}
	(z^2;q_1^2,\ldots q_N^2)_{\infty}=(z;q_1,\ldots q_N)_{\infty}(-z;q_1,\ldots q_N)_{\infty},\label{sqmult} 
	\end{equation}
	which follows from the splitting of each multiplier $\left(1-z^2\prod\limits_{k=1}^Nq_k^{2i_k}\right)=
	\left(1-z\prod\limits_{k=1}^Nq_k^{i_k}\right)\left(1+z\prod\limits_{k=1}^Nq_k^{i_k}\right)$.
	
	The $q$-Pochhammer symbol $(z;q_1,\ldots q_N)_{\infty}$ can be defined for the other region of parameters $q_k$
	using the formula
	\begin{multline}
	(z;q_1,\ldots q_N)_{\infty}=\exp\left(\sum_{i_1,\ldots i_N=0}^{\infty}\log\left(1-z\prod_{k=1}^Nq_k^{i_k}\right)\right)=
	\exp\left(-\sum_{i_1,\ldots i_N=0}^{\infty}\sum_{m=1}^{\infty}\frac{z^m}m\prod_{k=1}^Nq_k^{mi_k}\right)=\\
	=\exp\left(-\sum_{m=1}^{\infty}\frac{z^m}m\prod_{k=1}^N\frac1{1-q_k^m}\right).
	\label{Pochhammer_exp}
	\end{multline}
	The exponent expression converges in the region $|q_k|\neq 1$ but with a new additional requirement $|z|<1$.
	From this exponent definition we could define $q$-Pochhammer symbol in the region with some $|q_k|>1$ and $|z|<1$.
	Then in the appropriate region (i.e. $|z|,|zq_1|<1$) we have a relation
	\begin{equation}
	(z;q_1^{-1},q_2,\ldots q_N)_{\infty}=(zq_1;q_1,\ldots q_N)^{-1}_{\infty}. \label{qtrans}
	\end{equation}
	We can analytically continue this relation via \eqref{Pochhammer_def} to an arbitrary value of $z$.
	
	This paper uses only $N=1,2$ $q$-Pochhammer symbols.
	
	We also use a combination of $N=1$ $q$-Pochhammer symbols given by the $\theta$-function
	\begin{equation}\label{eq:theta}
	\theta(z;q)=(z;q)_{\infty}(qz^{-1};q)_{\infty}=\frac{1}{(q;q)_{\infty}}\sum_{k\in\mathbb{Z}}(-1)^k q^{\frac{k(k-1)}{2}}z^k,
	\end{equation}
	where the last equality is the Jacobi triple product.
	It follows from the definition and \eqref{qshift} that the $\theta$- function satisfies
	\begin{equation}
	\theta(qz;q)=-z^{-1}\theta(z;q)=\theta(z^{-1};q). \label{tshift}
	\end{equation}
	From the continuation \eqref{qtrans} we could define
	\begin{equation}
	\theta(z;q^{-1})=\theta^{-1}(qz;q). 
	\end{equation}
	The useful object is also the elliptic Gamma function, which is the combination of $N=2$ $q$-Pochhammer symbols
	\begin{equation}
	\Gamma(z;q_1,q_2)=\frac{(q_1q_2z^{-1};q_1,q_2)_{\infty}}{(z;q_1,q_2)_{\infty}}. \label{ellGamma}  
	\end{equation}
	It should not be confused with the ordinary Gamma or Barnes $\mathsf{G}$-function,
	its trigonometric or multiple analogs, the last are defined in the next paragraph.
	 Elliptic Gamma function satisfy relations
	 \begin{equation}
	 \Gamma(q_1z;q_1,q_2)=\theta(z;q_2) \Gamma(z;q_1,q_2) ,\quad
	 \Gamma(q_2z;q_1,q_2)=\theta(z;q_1) \Gamma(z;q_1,q_2),
	 \end{equation}
	 which follow from the definition, so as the useful relation
	 \begin{equation}\label{eq:ellGamma_shift}
	    \frac{\Gamma(uq^2;q,q)\Gamma(u;q,q)}{\Gamma(uq;q,q)^2}=-u^{-1}. 
	 \end{equation}
	
	\paragraph{Multiple gamma functions.}
	Here we collect some facts about multiple Gamma functions.
	These functions are in some sense $q\rightarrow1$ limit of
	$q$-Pochhammer symbols and admit many analogous properties.
	
	Following \cite[App. E]{NY03L}, introduce the function
	\begin{equation}
	\gamma_{\epsilon_1,\epsilon_2}(x;\Lambda):=
	\frac{d}{ds}|_{s=0}\frac{\Lambda^s}{\Gamma(s)}\int_{0}^{+\infty}\frac{dt}{t} t^s \frac{e^{-tx}}{(e^{\epsilon_1t}-1)(e^{\epsilon_2t}-1)}. \label{gammaNYdef}
	\end{equation}
	This integral converges at $t=0$ for $\operatorname{Re}s>2$. The analytic continuation is done by standard methods, see below. Also it is necessary to assume $\operatorname{Re}x>0$ and $\operatorname{Re}\epsilon_1,\operatorname{Re}\epsilon_2\neq 0$ for convergence. This function can be analytically continued as a function of $x$
	
	This function is homogeneous \begin{equation}\label{eq:homog:eps12}
	\gamma_{\epsilon_1,\epsilon_2}(x;\Lambda)=\gamma_{\epsilon_1/\alpha,\epsilon_2/\alpha}(x/\alpha;\Lambda/\alpha),
	\end{equation}
	if $\operatorname{Re}(\alpha/\epsilon_i) \operatorname{Re}\epsilon_i>0$, $i=1,2$.

	We can remove its dependence on $\Lambda$ by the calculation
	\begin{equation}
	\begin{aligned}
	&\gamma_{\epsilon_1,\epsilon_2}(x;\Lambda)-\gamma_{\epsilon_1,\epsilon_2}(x;1)=
	\log\Lambda \lim_{s\rightarrow0}\Gamma^{-1}(s)\int_0^{+\infty}\frac{dt}{t} t^s \frac{e^{-tx}}{(e^{\epsilon_1t}-1)(e^{\epsilon_2t}-1)}=\\
	&=\frac{\log\Lambda}{\epsilon_1\epsilon_2} \lim_{s\rightarrow0}\Gamma^{-1}(s)\int_0^{+\infty}\frac{dt}{t} t^s 
	e^{-tx}\left(t^{-2}-\frac12(\epsilon_1+\epsilon_2)t^{-1}+\frac{1}{12}(\epsilon_1^2+\epsilon_2^2+3\epsilon_1\epsilon_2)+O(t)\right)=\\
	&=\frac{\log\Lambda}{\epsilon_1\epsilon_2} \lim_{s\rightarrow0} \frac{\Gamma(s-2)x^{2-s}-\frac12(\epsilon_1+\epsilon_2)\Gamma(s-1)x^{1-s}+
		\frac{1}{12}(\epsilon_1^2+\epsilon_2^2+3\epsilon_1\epsilon_2)\Gamma(s)x^{-s}}{\Gamma(s)}=\\
	&=\frac{\log\Lambda}{\epsilon_1\epsilon_2}\left(\frac12x^2+\frac12(\epsilon_1+\epsilon_2)x+\frac{1}{12}(\epsilon_1^2+\epsilon_2^2+3\epsilon_1\epsilon_2)\right).
	\label{extr_dep_2}
	\end{aligned}
	\end{equation}
	
	We can introduce also its $1$-parameter analog $\gamma_{\epsilon}$ 
	\begin{equation}
	\gamma_{\epsilon}(x;\Lambda):=
	\frac{d}{ds}|_{s=0}\frac{\Lambda^s}{\Gamma(s)}\int_{0}^{+\infty}\frac{dt}{t} t^s \frac{e^{-tx}}{e^{\epsilon t}-1},\quad \operatorname{Re}x>0.
	\label{gammaNY1def}
	\end{equation}
	This function is also homogeneous (for $\operatorname{Re}(\alpha/\epsilon) \operatorname{Re}\epsilon>0$)
	\begin{equation}\label{eq:homog:eps}
	\gamma_{\epsilon}(x;\Lambda)=\gamma_{\epsilon/\alpha}(x/\alpha;\Lambda/\alpha),
	\end{equation}
	and satisfies
	\begin{equation}
	\begin{aligned}
	&\gamma_{\epsilon}(x;\Lambda)-\gamma_{\epsilon}(x;1)=
	\log\Lambda \lim_{s\rightarrow0}\Gamma^{-1}(s)\int_0^{+\infty}\frac{dt}{t} t^s \frac{e^{-tx}}{e^{\epsilon t}-1}=\\
	&=\frac{1}{\epsilon}\log\Lambda \lim_{s\rightarrow0}\Gamma^{-1}(s)\int_0^{+\infty}\frac{dt}{t} t^s 
	e^{-tx}\left(t^{-1}-\epsilon/2+O(t)\right)=\\
	&=\frac{1}{\epsilon}\log\Lambda \lim_{s\rightarrow0} \frac{\Gamma(s-1)x^{1-s}-\frac12\epsilon\Gamma(s)x^{-s}}{\Gamma(s)}=
	-\frac1\epsilon(x+\epsilon/2)\log \Lambda.\label{extr_dep_1}
	\end{aligned}
	\end{equation}
	
	The function $\exp(\gamma_{\epsilon_1,\epsilon_2}(x;1))$ is symmetric with respect to the $\epsilon_1$, $\epsilon_2$. Functions $\exp(\gamma_{\epsilon_1,\epsilon_2}(x;1))$ and $\exp(\gamma_{\epsilon}(x;1))$  satisfy analogs of the relations
	\eqref{qshift}, \eqref{permult}, \eqref{qtrans}:
	\begin{eqnarray}
	\exp{\gamma_{\epsilon_1,\epsilon_2}(x;1)}&=&\exp{\gamma_{\epsilon_2}(x+\epsilon_1;1)} \exp{\gamma_{\epsilon_1,\epsilon_2}(x+\epsilon_1;1)},\label{gammashift}\\
	\exp{\gamma_{\epsilon_1,\epsilon_2}(x;1)}&=&\exp{\gamma_{2\epsilon_1,\epsilon_2}(x;1)}\exp{\gamma_{2\epsilon_1,\epsilon_2}(x-\epsilon_1;1)},\label{gammapermult}\\
	\exp{\gamma_{-\epsilon_1,\epsilon_2}(x;1)}&=&\exp(-\gamma_{\epsilon_1,\epsilon_2}(x-\epsilon_1;1)),\label{gammatrans}
	\end{eqnarray}
	and 
	\begin{eqnarray}
	\exp{\gamma_{\epsilon}(x;1)}&=&\frac{1}{x+\epsilon} \exp{\gamma_{\epsilon}(x+\epsilon;1)},\label{gamma1shift}\\
	\exp{\gamma_{\epsilon}(x;1)}&=&\exp{\gamma_{2\epsilon}(x;1)}\exp{\gamma_{2\epsilon}(x-\epsilon;1)},\label{gamma1permult}\\
	\exp{\gamma_{-\epsilon}(x;1)}&=&\exp(-\gamma_{\epsilon}(x-\epsilon;1)).\label{gamma1trans}
	\end{eqnarray}
	These relations follow directly from the definition \eqref{gammaNYdef}.
	We also have relations
	\begin{eqnarray}
	\exp{\gamma_{\epsilon_1/\alpha,\epsilon_2/\alpha}(x/\alpha;1)}=\alpha^{\frac{x^2+(\epsilon_1+\epsilon_2)x+\frac{1}{6}(\epsilon_1^2+\epsilon_2^2+3\epsilon_1\epsilon_2)}{2\epsilon_1\epsilon_2}}\exp{\gamma_{\epsilon_1,\epsilon_2}(x;1)},\label{gammascal}\\
	\exp{\gamma_{\epsilon/\alpha}(x/\alpha;1)}=\alpha^{-x/\epsilon-1/2}\exp{\gamma_{\epsilon}(x;1)}\label{gamma1scal},
	\end{eqnarray}
	which follow from \eqref{extr_dep_2}, 
	\eqref{extr_dep_1} and homogeneity \eqref{eq:homog:eps12}, \eqref{eq:homog:eps} (under the corresponding restrictions).
	
	In the case $\operatorname{Re}x>0, \operatorname{Re}\epsilon_1,\epsilon_2,\epsilon<0$ we can rewrite definitions \eqref{gammaNYdef}, \eqref{gammaNY1def} 
	\begin{align}
	\exp(\gamma_{\epsilon_1,\epsilon_2}(x;1))&=\exp\left(\frac{d}{ds}|_{s=0}\sum_{m,n=0}^{+\infty}(x-\epsilon_1m-\epsilon_2n)^{-s}\right), \label{zeta2}\\
	\exp(-\gamma_{\epsilon}(x;1))&=\exp\left(\frac{d}{ds}|_{s=0}\sum_{n=0}^{+\infty}(x-\epsilon n)^{-s}\right). 
	\end{align}
	These formulas could be used for analytic continuation of given functions on the region $\operatorname{Re}x<0$ similar to \eqref{Pochhammer_exp}. Moreover, the latter expressions are just multiple Gamma functions, so we have (under condition $\operatorname{Re}\epsilon_1,\epsilon_2,\epsilon<0$)
	\begin{align}
	\exp(\gamma_{\epsilon_1,\epsilon_2}(x;1))=
	\Gamma_2(x;-\epsilon_1,-\epsilon_2),  \label{gamma2}\\
	\exp(-\gamma_{\epsilon}(x;1))=\Gamma_1(x;-\epsilon)
	=\frac{(-\epsilon)^{-x/\epsilon-1/2}}{\sqrt{2\pi}}\Gamma(-x/\epsilon)\label{gamma1}. 
	\end{align}
	
	\begin{Remark}
		There is a difference between the functions $\gamma_{\epsilon}, \gamma_{\epsilon_1,\epsilon_2}$ and multiple gamma functions $\Gamma_1$, $\Gamma_2$. 
		The latter functions are analytic in $\epsilon_i$ everywhere except zeroes, but the functions~$\gamma$ are not defined for purely imaginary values of $\epsilon_i$ (see \eqref{gammaNY1def} and \eqref{gammaNYdef} and compare with the requirement $|q_i|\neq1$ for \eqref{Pochhammer_exp}).
		Therefore, the connection depends on  $\operatorname{Re}\epsilon_{1,2},\epsilon\lessgtr0$. For example, from \eqref{gamma1trans} and \eqref{gamma1}, we obtain for $\operatorname{Re}\epsilon>0$ that 
		$\exp(-\gamma_{\epsilon}(x;1))=\epsilon^{-x/\epsilon-1/2}\frac{\sqrt{2\pi}}{\Gamma(1+x/\epsilon)}$.
    \end{Remark}
	
	\section{Perturbative part of the Nekrasov function}
	\label{sec:perturb}
	\paragraph{Perturbative part of the 5d Nekrasov function.}
	We consider the so-called perturbative part of the 5d Nekrasov function according to \cite{NY05} \cite[Section 4.2]{NY05} for the rank $2$ case.
	The full partition function is $\mathcal{Z}_{NY}=\mathcal{Z}_{pert}\mathcal{Z}_{inst}$ and perturbative term is given by
	\begin{equation}
	\mathcal{Z}_{pert}(a;\epsilon_1, \epsilon_2; R|\Lambda)=\exp(-\td{\gamma}_{\epsilon_1,\epsilon_2;R}(a|\Lambda)-\td{\gamma}_{\epsilon_1,\epsilon_2;R}(-a|\Lambda)),  
	\end{equation}
	where
	\begin{equation}
	\begin{aligned}
	&\td{\gamma}_{\epsilon_1,\epsilon_2;R}(x|\Lambda)=\frac{1}{2\epsilon_1\epsilon_2}\left(x^2+(\epsilon_1+\epsilon_2)x+\frac{\epsilon_1^2+\epsilon_2^2+3\epsilon_1\epsilon_2}{6}\right)\log(\Lambda)+\\
	&+\frac{1}{\epsilon_1\epsilon_2}\left(-\frac{R}{12}(x+(\epsilon_1+\epsilon_2)/2)^3+\frac{\pi^2}{6R}(x+(\epsilon_1+\epsilon_2)/2)-\frac{\zeta(3)}{R^2}\right)
	+\sum_{n=1}^{+\infty}\frac1n\frac{e^{-R nx}}{(e^{R n \epsilon_1}-1)(e^{R n \epsilon_2}-1)}.
	\end{aligned}
	\end{equation}
	
	Now we relate $\mathcal{Z}_{pert}$ with $\mathcal{Z}_{cl}$ and $\mathcal{Z}_{1-loop}$ from \eqref{Zstr}. The aim is to obtain formulas \eqref{5d_def} for $\mathcal{Z}_{cl}$ and $\mathcal{Z}_{1-loop}$
	from the above formula for $\mathcal{Z}_{pert}$. Simplifying the expression for $\mathcal{Z}_{pert}$ and using \eqref{Pochhammer_exp}, we obtain
	\begin{equation}
	\mathcal{Z}_{pert}(a;\epsilon_1, \epsilon_2; R|\Lambda)=C(\epsilon_1,\epsilon_2;R|\Lambda)\times
	(\Lambda e^{-\frac{R}4(\epsilon_1+\epsilon_2)})^{-\frac{a^2}{\epsilon_1\epsilon_2}}\times(e^{-Ra};q_1,q_2)_{\infty}(e^{Ra};q_1,q_2)_{\infty},
	\end{equation}
	where
	\begin{equation}
	C(\epsilon_1,\epsilon_2;R|\Lambda)=e^{\frac{2\zeta(3)}{\epsilon_1\epsilon_2R^2}-\frac{\pi^2}{6R}\frac{\epsilon_1+\epsilon_2}{\epsilon_1\epsilon_2}+R\frac{(\epsilon_1+\epsilon_2)^3}{48\epsilon_1\epsilon_2}} 
	\Lambda^{-\frac16(1+\frac{(\epsilon_1+\epsilon_2)^2}{\epsilon_1\epsilon_2})}.
	\end{equation}
	The second and third multipliers, separated by $\times$, coincide with  $\mathcal{Z}_{cl}$ and $\mathcal{Z}_{1-loop}$ 
	correspondingly (see \eqref{5d_def}).
	
	In this paper, our main object of interest is Nakajima-Yoshioka blowup relations which have the structure (\cite[(2.5)-(2.7)]{NY05})
	\begin{equation}
	\beta_d\mathcal{Z}_{NY}(a;\epsilon_1,\epsilon_2;R|\Lambda)=\widehat{\mathcal{Z}}_d(a;\epsilon_1,\epsilon_2;R|\Lambda),
	\end{equation}
	where $d=0,1,2$, $\beta_d$ are some factors,  and (\cite[(4.14)]{NY05})
	\begin{multline}
	\widehat{\mathcal{Z}}_d(a;\epsilon_1,\epsilon_2;R|\Lambda)=e^{\frac{2d-1}{24}R(\epsilon_1+\epsilon_2)} \sum_{n}\mathcal{Z}_{NY}(a+2\epsilon_1n;\epsilon_1,\epsilon_1-\epsilon_2;R|e^{\epsilon_1R(d-1)}\Lambda)
	\\ \times \mathcal{Z}_{NY}(a+2\epsilon_2n;\epsilon_1-\epsilon_2,\epsilon_2;R|e^{\epsilon_2R(d-1)}\Lambda),  
	\end{multline}
	We use another normalization $\mathcal{Z}$ without the factor $C(\epsilon_1,\epsilon_2;R|\Lambda)$:
	$\mathcal{Z}_{NY}=C\mathcal{Z}$. Since  $C(\epsilon_1,\epsilon_2;R|\Lambda)$ does not depend on $a$, it will not change substantially the structure of
	Nakajima-Yoshioka blowup relations.
	Moreover, since 
	\begin{equation}
	C(\epsilon_1,\epsilon_2;R|\Lambda)=e^{\frac{2d-1}{24}R(\epsilon_1+\epsilon_2)}C(\epsilon_1,\epsilon_2-\epsilon_1;R|e^{\frac14\epsilon_1R(d-1)}\Lambda)
	C(\epsilon_1-\epsilon_2,\epsilon_2;R|e^{\frac14\epsilon_2R(d-1)}\Lambda), 
	\end{equation}
	removing of the term $C$ in the $\mathcal{Z}_{pert}$ just removes the term $e^{\frac{2d-1}{24}R(\epsilon_1+\epsilon_2)}$ in the Nakajima-Yoshioka blowup relations.
	So we obtain blowup relations \eqref{qNY1}, \eqref{qNY2}, \eqref{qNY3} on $\mathcal{Z}$, the factors $\beta_d$ are given in the l.h.s. of these relations.
	
	
	\paragraph{Perturbative part of the 4d Nekrasov function.}
	Perturbative part of the 4d Nekrasov function is given by \cite{NY03L}
	\begin{equation}
	\mathcal{Z}_{pert}(a;\epsilon_1, \epsilon_2|\Lambda)=\exp(-\gamma_{\epsilon_1,\epsilon_2}(a;\Lambda)-\gamma_{\epsilon_1,\epsilon_2}(-a;\Lambda)), 
	\end{equation}
	where the function $\gamma_{\epsilon_1,\epsilon_2}$ is given by \eqref{gammaNYdef}.
	
	Extracting the dependence of the perturbative term on $\Lambda$ with the help of \eqref{extr_dep_2}, we obtain
	\begin{equation}
	\mathcal{Z}_{pert}(a;\epsilon_1, \epsilon_2|\Lambda)=C(\epsilon_1,\epsilon_2|\Lambda)\times\Lambda^{-\frac{a^2}{\epsilon_1\epsilon_2}}
	\times\exp(-\gamma_{\epsilon_1,\epsilon_2}(a;1)-\gamma_{\epsilon_1,\epsilon_2}(-a;1)).
	\end{equation}
	The second and third multipliers coincide with $\mathcal{Z}_{cl}$ and $\mathcal{Z}_{1-loop}$ from \eqref{4d_def}. 
	The first multiplier
	\begin{equation}
	C(\epsilon_1,\epsilon_2|\Lambda)=\Lambda^{-\frac{1}{6\epsilon_1\epsilon_2}(\epsilon_1^2+\epsilon_2^2+3\epsilon_1\epsilon_2)}
	\end{equation}
	could be omitted just analogously to 5d case, because it satisfies
	\begin{equation}
	C(\epsilon_1,\epsilon_2|\Lambda)=C(\epsilon_1,\epsilon_2-\epsilon_1|\Lambda)C(\epsilon_1-\epsilon_2,\epsilon_2|\Lambda),
	\end{equation}
	and $\Lambda$ is not shifted in differential Nakajima-Yoshioka blowup relations \cite{NY03L} in the term $C(\epsilon_1,\epsilon_2|\Lambda)$, see  \cite[Sec. 4.5]{NY03L}. 
	
	\paragraph{Perturbative part of the 4d Nekrasov function: $c=1$ and $c=-2$ specialization.}
	Now consider special cases of $\epsilon_1, \epsilon_2$ we are mostly interested in this paper.
	Namely, $\epsilon_2=-\epsilon_1$ corresponding to the $c=1$ in the CFT notations 
	and $\epsilon_2=-2\epsilon_1$ corresponding to the $c=-2$ (see first formula from \eqref{cD}).
	We rewrite in these cases $1$-loop factor in terms of ordinary Barnes $\mathsf{G}$-function (as in  original formula for tau function from \cite{GIL12}).
	From \cite[Sec. 5]{Ad01} we have
	\begin{equation}
	\Gamma_2(x;1,1)=\mathsf{G}^{-1}(x)(2\pi)^{z/2-1}e^{\zeta'(-1)}. 
	\end{equation}
	This means that for $c=1$ $\mathcal{Z}_{cl}=z^{\sg^2}$ and up to a constant
	\begin{equation}
	\mathcal{Z}_{1-loop}=\frac{1}{\mathsf{G}(1+2\sg)\mathsf{G}(1-2\sg)},\label{loopG} 
	\end{equation}
	where we used \eqref{gammascal} and \eqref{gammatrans}.
	Analogously for $c=-2$ we have that $\mathcal{Z}_{cl}=(z/4)^{\sg^2/2}$ and up to the constant
	\begin{equation}
	\begin{aligned}
	\mathcal{Z}^+_{1-loop}=\frac{1}{\mathsf{G}(1/2+\sg)\mathsf{G}(1+\sg)\mathsf{G}(1/2-\sg)\mathsf{G}(1-\sg)},\\
	\mathcal{Z}^-_{1-loop}=\frac{1}{\mathsf{G}(1+\sg)\mathsf{G}(3/2+\sg)\mathsf{G}(1-\sg)\mathsf{G}(3/2-\sg)},
	\label{loop2G}
	\end{aligned}
	\end{equation}
	where we additionaly used \eqref{gammapermult}.
	\begin{Remark}
		Note that the constants we may have missed when expressing the perturbative term
		in terms of ordinary Barnes $\mathsf{G}$-function are not so harmless, since they should
		be consistent with the Nakajima-Yoshioka blowup relation. However, for given expressions \eqref{loopG}, \eqref{loop2G}
		the relative constant is given by the Barnes $\mathsf{G}$-function multiplication formula.
	\end{Remark}


	\noindent \textsc{Landau Institute for Theoretical Physics, Chernogolovka, Russia,\\
		Center for Advanced Studies, Skolkovo Institute of Science and Technology, Moscow, Russia,\\
		National Research University Higher School of Economics, Moscow, Russia,\\
		Institute for Information Transmission Problems, Moscow, Russia,\\
		Independent University of Moscow, Moscow, Russia}
	
	\emph{E-mail}:\,\,\textbf{mbersht@gmail.com}\\

	\noindent\textsc{National Research University Higher School of Economics, Moscow, Russia\\
		Center for Advanced Studies, Skolkovo Institute of Science and Technology, Moscow, Russia}
	
	\emph{E-mail}:\,\,\textbf{shch145@gmail.com}
\end{document}